\documentclass[11pt]{article}
\pdfoutput=1

\usepackage[usenames,dvipsnames]{xcolor}
\usepackage{amsfonts}
\usepackage{amsmath,amsthm,amssymb,dsfont}
\usepackage{enumerate}
\usepackage[english]{babel}
\usepackage{graphicx}	
\usepackage[caption=false]{subfig}
\usepackage[margin=2cm]{geometry}
\usepackage{url}
\usepackage{bbm}
\usepackage{booktabs,tabularx}
\usepackage[normalem]{ulem}
\usepackage{caption}

\usepackage{tikz}
\usetikzlibrary{shadows}
\definecolor{myyellow}{RGB}{242,226,149}
\definecolor{grey}{RGB}{150,150,150}
\definecolor{myblue}{RGB}{200,220,230}

\usetikzlibrary{chains}
\usetikzlibrary{fit}
\usepackage{pgflibraryarrows}		
\usepackage{pgflibrarysnakes}		
\usetikzlibrary{shapes.symbols,patterns} 

\usepackage{pgfplots}
\pgfplotsset{compat=newest}
\usepgfplotslibrary{fillbetween}

\usepackage{mathrsfs}

\usepackage{hyperref}
\hypersetup{colorlinks=true,citecolor=blue,linkcolor=blue,filecolor=blue,urlcolor=blue,breaklinks=true}

\usepackage{nicefrac}
\usepackage{mathtools}

\newtheorem{theorem}{Theorem}[section]
\newtheorem{lemma}[theorem]{Lemma}

\newtheorem{proposition}[theorem]{Proposition}

\newtheorem{definition}[theorem]{Definition}
\theoremstyle{remark}
\newtheorem{remark}[theorem]{Remark}

\newcommand{\comment}[1]{}

\newcommand*{\cA}{\mathcal{A}}
\newcommand*{\cB}{\mathcal{B}}

\newcommand*{\cG}{\mathcal{G}}
\newcommand*{\cH}{\mathcal{H}}

\newcommand*{\RR}{\mathbb{R}}

\newcommand*{\CC}{\mathbb{C}}

\newcommand*{\NN}{\mathbb{N}}
\newcommand*{\ZZ}{\mathbb{Z}}

\newcommand*{\eps}{\varepsilon}

\DeclareMathOperator{\tr}{tr}

\newcommand{\1}{\mathbf{1}}

\newcommand*{\id}{I}

\newcommand*{\ket}[1]{| #1 \rangle}
\newcommand*{\bra}[1]{\langle #1 |}

\newcommand{\ketbra}[2]{|#1\rangle\!\langle #2|}
\newcommand{\proj}[1]{|#1\rangle\!\langle #1|}

\newcommand*{\<}{\langle}
\renewcommand*{\>}{\rangle}

\newcommand{\edit}[1]{{#1}}

\newcommand{\psd}{\succeq}

\definecolor{mycolor1}{rgb}{0.00000,0.44700,0.74100}%
\definecolor{mycolor2}{rgb}{0.85000,0.32500,0.09800}%
\definecolor{mycolor3}{rgb}{0.92900,0.69400,0.12500}%
\definecolor{mycolor4}{rgb}{0.49400,0.18400,0.55600}%
\definecolor{lightgrey}{gray}{0.6}

\renewcommand{\H}{\mathbf{H}}

\newcommand{\im}{\text{im}}

\newcommand{\tomega}{\tilde{\omega}}
\newcommand{\nsd}{\preceq}

\newcommand{\cAloc}{\cA_{\text{loc}}}

\newcommand{\KMS}{\cG}
\newcommand{\KMSTI}{\cG^{\mathrm{TI}}}
\newcommand{\comm}{\text{comm}}
\newcommand{\bdex}{\partial^{ex}}

\newcommand{\OminNTI}{\<O\>^{\min}}
\newcommand{\OmaxNTI}{\<O\>^{\max}}

\newcommand{\Omin}{\<O\>^{\min,\mathrm{TI}}}
\newcommand{\Omax}{\<O\>^{\max,\mathrm{TI}}}

\newcommand{\OTI}{\<O\>^{\mathrm{TI}}}

\newcommand{\obsmin}[1]{\<#1\>^{\min}}
\newcommand{\obsmax}[1]{\<#1\>^{\max}}

\newcommand{\Pmin}{\mathsf{P}^{\min}}
\newcommand{\Pmax}{\mathsf{P}^{\max}}
\newcommand{\Pminp}{\mathsf{P}^{'\min}}

\newcommand{\pstar}{\mathsf{p}^{*}}

\newcommand{\pmin}{\mathsf{p}^{\min}}
\newcommand{\pmax}{\mathsf{p}^{\max}}
\newcommand{\pminTI}{\mathsf{p}^{\min,\mathrm{TI}}}
\newcommand{\pmaxTI}{\mathsf{p}^{\max,\mathrm{TI}}}

\newcommand{\pminp}{\mathsf{p}^{'\min}}
\newcommand{\pmaxp}{\mathsf{p}^{'\max}}

\newcommand{\sites}{\Gamma}

\newcommand{\suppref}[2]{{#1}}
\newcommand{\mainref}[2]{{#1}}
\usepackage{authblk}

\title{Certified algorithms for equilibrium states \\
of local quantum Hamiltonians}
\author[1]{Hamza Fawzi}
\author[2]{Omar Fawzi}
\author[1]{Samuel O. Scalet\footnote{Email: sos25@cam.ac.uk}}
\affil[1]{Department of Applied Mathematics and Theoretical Physics, University of Cambridge, United Kingdom}
\affil[2]{{Univ Lyon, Inria, ENS Lyon, UCBL, LIP, Lyon, France}}
\date{}

\begin{document}

\maketitle

\begin{abstract}

Predicting observables in equilibrium states is a central yet notoriously hard question in quantum many-body systems.
In the physically relevant thermodynamic limit, certain mathematical formulations of this task have even been shown to result in undecidable problems.
Using a finite-size scaling of algorithms devised for finite systems often fails due to the lack of certified convergence bounds for this limit.
In this work, we design certified algorithms for computing expectation values of observables in the equilibrium states of local quantum Hamiltonians, both at zero and positive temperature. Importantly, our algorithms output rigorous lower and upper bounds on these values.
This allows us to show that expectation values of local observables can be approximated in finite time, contrasting related undecidability results.
When the Hamiltonian is commuting on a 2-dimensional lattice, we prove fast convergence of the hierarchy at high temperature and as a result for a desired precision $\eps$, local observables can be approximated by a convex optimization program of quasi-polynomial size in $1/\eps$.

\end{abstract}

\section{Introduction}

A central question in physics is to determine the properties of a many-body quantum system as a function of the interaction between its constituents.
The topic of Hamiltonian complexity~\cite{osborne2012hamiltonian,gharibian2015quantum} studies this question from a complexity-theoretic point of view. The results in Hamiltonian complexity suggest that efficient algorithms answering this question are unlikely to exist. In fact, determining the ground energy is hard for the complexity class QMA~\cite{kempe2005}, and this even holds for translation-invariant systems~\cite{gottesman2010quantum}. This means that we do not expect polynomial time (classical or quantum) algorithms for this problem. Computing expectation values of observables in the ground state is even harder than the ground energy~\cite{ambainis2014physical}. In the thermodynamic limit when the number of systems is taken to infinity, finding good approximations to the energy can be hard for fixed Hamiltonians~\cite{aharonov2022hamiltonian,watson2022computational} and the spectral gap is even uncomputable~\cite{cubitt2015undecidability}.

These complexity results put severe limitations on provably efficient and correct classical and quantum algorithms for the fundamental questions in many-body physics. However, these limitations only apply to algorithms that have the required properties \emph{for all} valid instances of the problem. Moreover, typically the instances showing hardness are highly contrived. In order to avoid these limitations, we consider here \emph{certified algorithms} where we require the correctness condition for all instances, but relax the condition of provable efficiency for all instances.
\edit{We say that an algorithm for computing $\pstar(H)$ is \emph{certified} if on input $H$, it outputs a pair of numbers $(\pmin_{\ell}, \pmax_{\ell})$ such that we have $\pstar(H) \in [\pmin_{\ell}, \pmax_{\ell}]$ for all $H$ and all $\ell$, i.e., the algorithm provides upper and lower bounds on the quantity of interest.
Even without having any performance guarantees \emph{a priori}, the error is bounded by the size of the interval \emph{a posteriori}.}

We further require that as $\ell \to \infty$, the interval $[\pmin_{\ell}, \pmax_{\ell}]$ converges to the desired value $\pstar(H)$.
Here, $\ell$ is a parameter that governs the runtime of the algorithm such that fast convergence in $\ell$ leads to an efficient algorithm. The above cited complexity results show e.g., that if $\pstar(H)$ is the ground energy of $H$, then for small $\ell$, the interval $[\pmin_{\ell}, \pmax_{\ell}]$ has to be large for some $H$. However, for any input $H$ of interest, we can always run the algorithm and the returned information will be correct and certifiably so. As such, no a priori analysis of the convergence speed is needed in order to obtain rigorous approximations of the quantity of interest.

A well-studied general way of determining properties of thermal states is by preparing such states on a quantum computer. This includes for example works on quantum Metropolis sampling ~\cite{terhal2000problem,temme2011quantum} and more generally quantum Gibbs samplers~\cite{rall2023thermal,chen2023quantum}. Here, we focus on \emph{classical} algorithms, which clearly cannot prepare quantum thermal states, but can still compute specified properties of general quantum thermal states. We note that our algorithms have a similar feature to (quantum) Metropolis sampling, and more generally (quantum) Markov Chain Monte Carlo methods, in that they can be applied for any local Hamiltonian but this comes at the (unavoidable) cost of not always having fast convergence.

\edit{In this work, we consider a hierarchy of convex optimization problems that provide certified algorithms for computing observables in ground and thermal states in finite and infinite systems. As a main result, we prove decidability of a formulation of the equilbrium observable problem, which contrasts with previous closely related undecidability results. The main technical ingredient is a formulation of constraints relaxing the Gibbs condition in a convex and algorithmically feasible manner. Furthermore, we are able to prove efficiency in a more restricted setting. Preliminary numerical results also suggest that the method can be of use in practice.}

\paragraph{Results}
We now present the setup in an informal way. We refer to \suppref{Section~\ref{sec:preliminaries}}{the Supplementary information} for more formal statements. We assume that our system is described by a local Hamiltonian on a discrete set of sites  $\sites$, which can be formally written as
\begin{equation}
H = \sum_{X \subset \sites} h_X
\end{equation}
where $h_X$ is the interaction term acting only on the finite set of sites $X \subset \sites$. We assume that the Hilbert space describing each site has finite dimension $d$ and the Hamiltonian $H$ 
is local, meaning that $h_X$ is nonzero only when $X$ has size smaller than a constant. 
For finite systems, the set of equilibrium states at temperature $T \geq 0$ is the set of states $\rho$ that minimize $\tr(H\rho)-TS(\rho)$, where $\tr(H\rho)$ is the energy, and $S(\rho)$ is the entropy.
At $T > 0$, this set reduces to a single equilibrium state also called the Gibbs state given by $\rho = \frac{e^{-H/T}}{\tr e^{-H/T}}$. However, for infinite systems the thermal equilibrium state is not necessarily unique, and this property lies at the heart of the existence of phase transitions.

Given the description of such a system, our objective is to determine physical properties of the corresponding equilibrium states. Consider a local observable $O$, e.g., the magnetization in the z-direction for a spin-$1/2$ system for some site $x \in \sites$. We are interested in the set of values $\tr(\rho O)$ that an equilibrium state at temperature $T = 1/\beta$ can take: 
\begin{equation}
[\OminNTI_{\beta},\OmaxNTI_{\beta}] := \left\{\tr(\rho O) : \rho \text{ is an equilibrium state at temperature $T=1/\beta$} \right\}.
\end{equation}
Note that for a finite system and $\beta < \infty$, we have $\obsmin{O}_{\beta} = \obsmax{O}_{\beta}$ as the thermal state is unique. However, we can have $\obsmin{O}_{\beta} < \obsmax{O}_{\beta}$ for $\beta = \infty$ if the ground space of $H$ is degenerate and for infinite systems when there are multiple thermal equilibrium states.

\paragraph{Main result} We formulate a hierarchy of convex optimization programs which produce converging lower (resp. upper) bounds on $\obsmin{O}_{\beta}$ (resp. $\obsmax{O}_{\beta}$) for any $\beta \geq 0$. Our optimization problems are formulated in terms of the matrix-valued relative entropy function
\begin{equation}
\label{eq:mre-intro}
D_{op}(A\|B) = A^{1/2} \log(A^{1/2} B^{-1} A^{1/2}) A^{1/2}
\end{equation}
defined for arbitrary positive definite matrices $A,B$, and which is jointly convex in $(A,B)$. 
Consider a finite subset $\Lambda$ of the lattice sites $\Gamma$ containing the support of the local observable $O$. Then we can formulate the following convex optimization program over density operators supported on $\overline{\Lambda}=\Lambda \cup \bdex \Lambda$ where $\bdex \Lambda$ is the external boundary of $\Lambda$ (see \suppref{\eqref{eq:bdexLambda}}{the Supplementary information Equation (4)} for the precise definition):
\label{eq:optintro}
\begin{align}
\underset{\rho}{\text{min/max}} \quad & \tr(\rho O)\\
\text{s.t.} \quad & 
\rho \text{ density operator on } (\CC^d)^{\otimes |\overline{\Lambda}|}\\
& \tr_{\overline{\Lambda} \setminus \Lambda}(H_{\overline{\Lambda}} \rho - \rho H_{\overline{\Lambda}}) = 0\label{eq:optintro-comm}\\
& D_{op} \left( \, A_{\rho} \| B_{\rho} \, \right) \; \nsd \; \beta C_{\rho}.\label{eq:optintro-EEB}
\end{align}
Here,  $H_{\overline{\Lambda}}$ is the truncated Hamiltonian acting on $\overline{\Lambda}$, and $A_{\rho},B_{\rho}, C_{\rho}$ are Hermitian matrices that depend linearly on $\rho$, defined by:
\begin{equation}
(A_{\rho})_{ij} = \tr(\rho a_i^* a_j), \quad (B_{\rho})_{ij} = \tr(\rho a_j a_i^*), \quad (C_{\rho})_{ij} = \tr(\rho a_i^* [H_{\overline{\Lambda}}, a_j])
\end{equation}
where $\{a_i\}$ is a basis of the space of operators acting on $(\CC^d)^{\otimes |\overline \Lambda|}$.
Note that the program~\eqref{eq:optintro} involves a matrix variable of dimension $d^{|\overline{\Lambda}|} \times d^{|\overline{\Lambda}|}$ and convex constraints involving matrices of dimension at most $d^{2|\overline{\Lambda}|} \times d^{2|\overline{\Lambda}|}$.

By taking larger and larger subsets $\Lambda \uparrow \sites$, one can prove that the solutions of the convex optimization programs will converge to the expectation values $\OminNTI_{\beta}$ and $\OmaxNTI_{\beta}$. This is the content of the following theorem.

\begin{theorem}[Certified algorithms for expectation values of equilibrium states]
\label{thm:main}
Let $\Lambda_0 \subset \Lambda_1 \subset \dots \subset \sites$ be an increasing sequence such that the support of the local observable $O$ is contained in $\Lambda_0$.
For any $\ell \in \NN$, let $\pmin_{\ell}$ and $\pmax_{\ell}$ be respectively the minimum and maximum values of the convex optimization problems \eqref{eq:optintro} with $\Lambda = \Lambda_{\ell}$. Then we have
\begin{equation}
\pmin_{\ell} \leq \OminNTI_{\beta} \leq \OmaxNTI_{\beta} \leq \pmax_{\ell}.
\end{equation}
Furthermore, $\pmin_{\ell} \uparrow \OminNTI_{\beta}$ and $\pmax_{\ell} \downarrow \OmaxNTI_{\beta}$  as $\ell \to \infty$.
\end{theorem}
We note that for finite systems, convergence happens after a finite number of steps, namely for $\ell$ such that $\Lambda_{\ell} = \sites$. \edit{While the convergence result becomes trivial in the finite case, the formulation of a convex program containing constraints for the Gibbs state still adds a novel tool for the numerical treatment of finite-sized quantum many-body systems. The more interesting case for our results is however that of infinite systems.} Given a fixed computational budget, one can choose $\ell$ and run the corresponding program~\eqref{eq:optintro} and obtain a superset $[\pmin_{\ell}, \pmax_{\ell}]$ of the target interval $[\OminNTI_{\beta}, \OmaxNTI_{\beta}]$. In the cases where $\OminNTI_{\beta} = \OmaxNTI_{\beta}$, e.g., for finite systems at $T > 0$, one can also run the program~\eqref{eq:optintro} with increasing values of $\ell$ until $\pmax_{\ell} \leq \pmin_{\ell} + \eps$, where $\eps$ is some desired accuracy $\eps$. We note that, as previously mentioned, the problem of computing expectation values for equilibrium states is at least QMA-hard so we cannot hope to have fast convergence for all choices of $H$. But we stress that no a priori analysis of convergence speed is required to obtain some desired accuracy: as soon as we get $\pmax_{\ell} \leq \pmin_{\ell} + \eps$, an additive $\eps$ approximation is guaranteed for this instance. 
Numerical results illustrating this algorithm can be found in \suppref{Section~\ref{sec:numerics}}{the Supplementary information Section 3}.

\paragraph{Translation-invariant infinite systems}
Translation-invariant Hamiltonians on the infinite lattice $\sites = \ZZ^D$, i.e., satisfying $h_{X} = h_{X+x}$ for all $x \in \ZZ^D$, play a specifically important role in statistical physics, in particular for understanding phase transitions. For such systems, one often considers the expectation value \emph{per site} of an observable, for example the energy per particle also called the energy density. One way to define expectation values per site is to compute the observable $O$ on any fixed site on a \emph{translation-invariant} equilibrium state of $H$. More precisely, we define the average expectation value of observable $O$ per site to be the interval
 \begin{equation}
 [\Omin_{\beta},\Omax_{\beta}] := \left\{\tr(\rho O) : \rho \text{ is a translation-invariant equilibrium state at temperature $T=1/\beta$} \right\}.
 \end{equation}
By adding to the program~\eqref{eq:optintro} a translation-invariance constraint in $\overline \Lambda$ given by
\begin{align}
\label{eq:translation-inv-constraint}
\tr_{\Sigma^c}(\rho) = \tr_{(\Sigma+x)^{c}}(\rho) \qquad \text{ for all $\Sigma \subset \overline \Lambda$, $x \in \ZZ^d$ such that $ \Sigma + x \subset \overline \Lambda$} ,
\end{align}
Theorem~\ref{thm:main} can be adapted for translation-invariant states, as follows:
\begin{theorem}
\label{thm:mainTI}
Let $\Lambda_0 \subset \Lambda_1 \subset \dots \subset \sites$ be an increasing sequence such that the support of the local observable $O$ is contained in $\Lambda_0$.
For any $\ell \in \NN$, let $\pminTI_{\ell}$ and $\pmaxTI_{\ell}$ be respectively the minimum and maximum values of the convex optimization problems given by \eqref{eq:optintro} together with the constraint~\eqref{eq:translation-inv-constraint}. Then we have
\begin{equation}
\pminTI_{\ell} \leq \Omin_{\beta} \leq \Omax_{\beta} \leq \pmaxTI_{\ell}.
\end{equation}
Furthermore, $\pminTI_{\ell} \uparrow \Omin_{\beta}$ and $\pmaxTI_{\ell} \downarrow \Omax_{\beta}$  as $\ell \to \infty$.
\end{theorem}

One difficulty we would like to highlight regarding the thermodynamic limit concerns the definition of equilibrium states at any given temperature $T\geq 0$. It was actually shown in~\cite{bausch2021uncomputability} that computing local observables on the ground state of infinite systems is undecidable. This might seem to contradict our result. This is not the case however, due to a subtle point in the definition of ground-state observables in the infinite limit.
The authors of~\cite{bausch2021uncomputability} define those as limits of ground-state observables in finite systems, i.e., the observable is computed for the ground-state of a truncated Hamiltonian with open boundary condition $H_{\Lambda} = \sum_{X \subset \Lambda} h_X$, where $\Lambda \subset \ZZ^D$ is finite, and afterwards the limit $\Lambda \uparrow \ZZ^D$ is taken.
The \edit{problem} definition we use in our work, which is standard in the operator algebraic framework, is more general and in particular contains limits of ground-state observables of finite system Hamiltonians with \textit{any choice of boundary conditions}. \edit{To avoid confusion, it should be added that this does not mean that our algorithm can compute the value of an observable for a specific boundary condition. Instead, it gives an outer relaxation of the interval given by all boundary conditions. A boundary condition is not input to the algorithm.} It is simple to construct a Hamiltonian where fixing the boundary condition while taking the limit excludes natural ground states (see \suppref{Remark \ref{rem:gsthermodynamiclimit} and also Remark \ref{rem:gibbsthermodynamiclimit}}{Remark 1.1 and Remark 2.3 in the supplemental information} concerning thermal states of the 2D Ising model).
The comparison however illustrates that even the existence of any convergent algorithm for observables in ground states is far from obvious, which is what we achieve in Theorem \ref{thm:mainTI}.
Considering the operator algebraic definition of equilibrium states leads to the natural computational problem defined as follows.

\begin{definition}\label{def:eop}

\end{definition}
\begin{center}
    \begin{tabularx}{\columnwidth-2cm}{@{}lX@{}}
    \toprule
    \multicolumn{2}{c}{\textsc{Equilibrium Observable Problem}}\\
    \midrule
    \bfseries Input: & \multicolumn{1}{p{0.75\textwidth}}{Local dimension $d$, local Hamiltonian term $h$ with rational coefficients, local observable $O$ with rational coefficients, rational number $a$, temperature $\beta\in[0,+\infty]$} \\
    \bfseries Promise: & Either $\Omin_{\beta}>a$ or  $\Omax_{\beta}<a$\\
    \bfseries Question: & Output YES if $\Omin_{\beta}>a$, NO otherwise\\
    \bottomrule
    \end{tabularx}
\end{center}

Note that this is a promise problem where the objective is to decide if all the equilibrium states have an expectation value per site $>a$ or all the equilibrium states have an expectation value per site $<a$.
\begin{theorem}[Decidability of translation-invariant ground state and thermal observables]
\label{thm:mainDec}
The Equilibrium Observable Problem is decidable.
\end{theorem}

\paragraph{Main ingredients} 
 For clarity of the discussion, we focus here on finite systems. A key ingredient to obtain our algorithms is to use an operator algebraic characterization of equilibrium states. These are expressed solely in terms of $\tr(\rho b)$ for  operators $b$. For $T=0$, the common definition of an equilibrium state is a state supported on the eigenspace of $H$ with the minimum eigenvalue.  It turns out that an equivalent operator algebraic formulation of this condition is:
\begin{equation}
\label{eq:gs}
\tr(\rho a^* [H,a]) \geq 0 \;\; \forall a
\end{equation}
where $[H,a]$ is the commutator and $a$ is an arbitrary observable. Intuitively, the condition above expresses the fact that the energy of $\rho$ has to increase under any perturbation. There are two crucial facts about \eqref{eq:gs}: First, if $a$ is supported on a small set $\Lambda$ of sites then the condition~\eqref{eq:gs} only depends on $\tr(\rho b)$ for operators $b \in \overline{\Lambda}$, where we recall that $\overline{\Lambda} = \Lambda \cup \bdex \Lambda$. Second, the inequalities in~\eqref{eq:gs} can be concisely captured by a convex constraint involving the positivity of a Hermitian matrix that depends linearly on $\rho$. By restricting the operators $a$ in \eqref{eq:gs} to be supported on $\Lambda$, this leads to~\eqref{eq:optintro} for $T=0$.

At positive temperature $T=1/\beta > 0$, the situation is more complicated. As previously mentioned, the common definition of a 
thermal equilibrium state is given by $\rho = \frac{e^{-\beta H}}{\tr(e^{-\beta H})}$. It turns out that an equivalent operator algebraic formulation is via the Kubo-Martin-Schwinger (KMS) condition:
\begin{equation}
\label{eq:kms1}
\tr(\rho ba) = \tr(\rho a e^{-\beta H} b e^{\beta H}) \quad \forall a,b.
\end{equation}
In the classical case, when the Hamiltonian is diagonal in a basis $\{\ket{\sigma}\}_{\sigma}$, Equation \eqref{eq:kms1} reduces to 
 \begin{equation}
 \label{eq:kmsclassical}
 \tr(\rho \proj{\sigma}) = \tr(\rho \proj{\sigma'}) \exp( -\beta (\bra{\sigma}H\ket{\sigma} - \bra{\sigma'}H\ket{\sigma'})), 
 \end{equation}
 for all $\sigma, \sigma'$. 
 When $\sigma'$ is obtained from $\sigma$ by flipping a single spin, these are known as the spin-flip equations and are exploited in the bootstrap approach for the classical Ising model \cite{bootstrapisinglattice}.

 To obtain Theorem \ref{thm:main}, a natural idea (as was done for $T=0$) is to relax the set of $\beta$-KMS states and impose the condition \eqref{eq:kms1} for a subset of observables $a,b$ supported on some small $\Lambda$. However, the main obstruction one is faced with is that even if $a,b$ are local observables the expression $a e^{-\beta H} b e^{\beta H}$ is in general not local, except for commuting Hamiltonians.
 As such, even though \eqref{eq:kms1} form a set of linear equations on $\rho$, they involve the expectation of $\rho$ on nonlocal observables.
 We circumvent this issue by using another characterization of thermal equilibrium states via so-called Energy-Entropy Balance (EEB) inequalities \cite{arakisewell}: this is an infinite set of scalar convex inequalities, each indexed by an operator $a$, which carve out the set of $\beta$-KMS states. On the one hand these inequalities are better suited than the $\beta$-KMS condition because they preserve locality, i.e., they only require the expectation value of the state $\rho$ on a finite region around the support of $a$. On the other hand, a drawback of these inequalities is that there are infinitely many of them, and unlike the inequalities~\eqref{eq:gs}
 it is not possible to express them as a linear positive semidefinite constraint. A key ingredient to prove our theorem is to formulate a matrix generalization of such inequalities using the matrix-valued relative entropy function \eqref{eq:mre-intro}. We show that an infinite set of scalar Energy-Entropy Balance inequalities can be compactly formulated by a single nonlinear convex matrix inequality of the form $D_{op}(A_{\rho}\|B_{\rho}) \nsd \beta C_{\rho}$ for some suitable matrices $A_{\rho},B_{\rho},C_{\rho}$ that depend linearly on $\rho$ (see \eqref{eq:optintro-EEB}). To the best of our knowledge, this is the first application of the matrix relative entropy function in the design of convex relaxations in quantum information. 
Optimization problems involving this function can be solved using interior-point methods \cite{sc} or via semidefinite programming approximations \cite{logapprox}.

\paragraph{Quantifying convergence speed}
Theorems~\ref{thm:main} and~\ref{thm:mainTI} do not provide quantitative guarantees on the convergence speed. As previously mentioned, one cannot hope to prove general fast convergence guarantees as even the special case where the observable $O$ corresponds to energy is unlikely to have a polynomial-time quantum algorithm, even when $D=1$~\cite{gottesman2010quantum}. 

However, one can expect provably fast convergence for some classes of Hamiltonians. We illustrate this by showing exponential convergence in two regimes for which it is known that no phase transitions can occur. For translation-invariant Hamiltonians in the high-temperature regime and in one dimension at any nonzero temperature, the set of equilibrium states reduces to a singleton, and for any local observable $O$ we have $\Omin_{\beta} = \Omax_{\beta} = \OTI_{\beta}$. Assuming a commuting local Hamiltonian $H$, the theorem below shows exponential convergence to $\OTI_{\beta}$ in the level $\ell$ of the convex optimization hierarchy.
\begin{theorem}
[Quantitative convergence rate]
\label{thm:main_quantitative}
Consider a translation-invariant local Hamiltonian $H$ on $\Gamma = \ZZ^D$ with $D\leq 2$. Assume furthermore that $H$ is \emph{commuting}, i.e., $[h_X,h_Y] = 0$ for all $X,Y \subset \ZZ^D$. For $D=1$ and $\beta_1=\infty$ or for $D=2$ and some $\beta_1 > 0$, we have for all $0 \leq \beta < \beta_1$, and for any local observable $O$,  $\Omin_{\beta} = \Omax_{\beta} = \OTI_{\beta}$ and, using the same notation as in 
Theorem~\ref{thm:main}
\begin{equation}
\max\{ \pmax_{\ell} - \OTI_{\beta} , \OTI_{\beta} - \pmin_{\ell} \} \leq c_1 \|O\| e^{-c_2 \ell}
\end{equation}
for some constants $c_1,c_2 > 0$ depending on the dimension and the interaction.
\end{theorem}

This result addresses an open problem raised in~\cite{cho2023bootstrap} about the speed of convex optimization  hierarchies for the (classical) Ising model, in particular whether exponential convergence holds away from criticality. Theorem~\ref{thm:main_quantitative} establishes such a statement for the more general class of commuting local Hamiltonians.

\paragraph{Related work}
A special case of the problem considered in this work is when the observable $O$ is the energy. Certified algorithms do exist for the ground energy of local Hamiltonians: one can combine semidefinite programming relaxations~\cite{Mazziotti2006,Pironio2010,Baumgratz2012,Barthel2012} for lower bounds and variational methods such as tensor networks~\cite{bridgeman2017,cirac2021,white1992} for upper bounds. 
Using such two-sided bounds on the ground energy, the recent works~\cite{wang2023,han2020} obtain bounds for expectation values of local observables in the ground state, although no convergence guarantees are given. The approach based on imposing the additional constraint \eqref{eq:gs}, which leads to convergence guarantees, was proposed independently and concurrently in \cite{navascuesKKT}, where it was derived as a special case of a method to strengthen semidefinite relaxations for noncommutative polynomial optimization problems by incorporating optimality conditions as additional constraints. These papers however do not discuss the case of positive temperature since the equilibrium states are not characterized by a noncommutative polynomial optimization problem which is linear in the state.
One could try to use lower bounds for the free energy coming from convex relaxations~\cite{Poulin2011}, but it is not clear how to use such bounds for observables.
We note however that for classical systems, thermal observables can be obtained via the bootstrap approach~\cite{bootstrapisinglattice,poland2019conformal}, by directly imposing the KMS conditions \eqref{eq:kmsclassical}. The resulting hierarchies were shown to be asymptotically convergent \cite{cho2023bootstrap}, even if only a weaker set of constraints are imposed. For quantum systems, to the best of our knowledge, Theorem~\ref{thm:main} provides the first certified algorithms for general observables in thermal states.

Note that, in some restricted settings, such as 1D gapped systems~\cite{landau2013} at zero temperature or 1D systems at positive temperature~\cite{hastings2006solving,molnar2015approximating,kuwahara2021,fawzi2023}, there are provably efficient algorithms computing representations of equilibrium states and thus expectation values, but such algorithms are tailored to these settings. In addition, for arbitrary dimensions and high temperatures provably efficient algorithms for computing the free energy exist~\cite{kuwahara2020clustering,harrow2020,mann2021}, which can be used to compute observables~\cite{bravyi2021complexity}.

\paragraph{Outlook} 
The preliminary numerical experiments in \suppref{Section~\ref{sec:numerics}}{the supplemental information} demonstrate that the approach proposed in this paper is not only theoretical, and with additional efforts on the computational side, can play an important role alongside other classical algorithms. For example, because our algorithm produces certified bounds, it can be used to rigorously benchmark variational algorithms for quantum many-body systems \cite{wu2023variational}. In addition, one can easily identify several directions for improving the accuracy of the algorithm presented here:
First, one can exploit additional convex constraints  that are known to hold for marginals of equilibrium states. For example, one can add entropy-based inequalities as proposed in \cite{fawzi2023entropy}. Furthermore, other valid inequalities tailored to specific models can be added to the convex relaxation such as reflection positivity which was used in \cite{bootstrapisinglattice} for the Ising model. Second, it would be very interesting to combine the methods from this paper with variational methods (e.g., from tensor networks) to obtain more accurate bounds, such as the recent work on the ground energy problem  \cite{kull2022lower}. Another natural question is whether one can use our algorithms within hybrid classical-quantum algorithms for quantum simulation problems~\cite{zhang2022computing}.
On the analytical side, it remains open whether convergence guarantees for other classes of Hamiltonian can be proven. Promising candidates would be regimes in which other classical algorithms are efficient such as thermal states in 1D and at high-temperature  (for general noncommuting Hamiltonians) as well as gapped ground states in 1D.

\section{Preliminaries}
\label{sec:preliminaries}

We consider spin systems on a (potentially infinite) discrete set of sites $\sites$ and we adopt the operator algebraic point of view. We describe the setup briefly here, and we refer to \cite{nachtergaele2004quantum} or \cite[Section 6.2]{brbook} for more details.

For any site $x \in \sites$, the local Hilbert space $\cH_x \simeq \CC^d$ has dimension $d$. The Hilbert space associated to a finite subset $X \subset \sites$ is the tensor product $\cH_X = \otimes_{x \in X} \cH_x$, and we let $\cA_X = \cB(\cH_X)$ be the algebra of observables supported on $X$. 
Let 
\begin{equation}
\cAloc = \bigcup_{\substack{X \subset \sites\\ X \text{ finite}}} \cA_X
\end{equation}
be the set of local observables, and $\cA = \overline{\cAloc}$ be its completion with respect to the operator norm, i.e., $\cA$ is the $C^*$-algebra of quasi-local observables. (Obviously, when $\sites$ is finite then $\cAloc = \cA = \overline{\cA}$ is the full algebra of complex matrices of size $d^{|\sites|} \times d^{|\sites|}$.) A state is a linear functional $\omega:\cA \to \CC$ such that $\omega(1) = 1$, $\omega(a^*) = \overline{\omega(a)}$, and $\omega(a^* a) \geq 0$ for all $a \in \cA$.

\paragraph{Hamiltonians} Consider a Hamiltonian $H$, which can be formally written as
\begin{equation}
H = \sum_{\substack{X \subset \sites\\X \text{ finite}}} h_X
\end{equation}
where $h_X \in \cA_{X}$ are the local interaction terms. We assume that the Hamiltonian $H$ is local, i.e., there exists a constant $r$ such that $h_{X}$ is nonzero only for $X$ having size at most $r$, for every site $x \in \Gamma$, $|\{X \subset \Gamma : x \in X, h_{X} \neq 0\}| \leq r$, and that $\|h_X\| \leq 1$ for all $X$.
In the case of the infinite lattice $\sites = \ZZ^D$, we say that the Hamiltonian is translation-invariant if $h_{X+x} = \tau_x(h_X)$ for all $X \subset \ZZ^D$ and $x \in \ZZ^D$, where $\tau_x:\cA_X \to \cA_{X+x}$ is the translation operator by $x$.
If $\Lambda$ is a finite subset of $\sites$, we let
\begin{equation}
\label{eq:HLambda}
H_{\Lambda} = \sum_{X \subset \Lambda} h_X \in \cA_{\Lambda}
\end{equation}
and
\begin{equation}
\label{eq:HtildeLambda}
\tilde{H}_{\Lambda} = \sum_{X \cap \Lambda \neq \emptyset} h_X \in \cA_{\overline \Lambda}
\end{equation}
where $\overline \Lambda = \Lambda \cup \bdex \Lambda$, and $\bdex \Lambda$ is the external boundary of $\Lambda$, i.e., 
\begin{equation}
    \label{eq:bdexLambda}
    \bdex \Lambda = \{y \in \Lambda^c : \exists Y \subset \ZZ^D, h_Y \neq 0, y \in Y, \Lambda \cap Y \neq \emptyset\}.
\end{equation}
When the lattice $\sites$ is infinite, $H$ is not a well-defined element of the algebra $\cA$; however for any $a \in \cAloc$, we formally write $[H,a]$ for the well-defined element of $\cAloc$ that is equal to
\begin{equation}
[H,a] = [\tilde{H}_{\Lambda},a],
\end{equation}
when $a \in \cA_{\Lambda}$.

\paragraph{Ground states} A state $\omega:\cA\to \CC$ is called a \emph{ground state} of $H$ if
\begin{equation}
\label{eq:groundstate}
\omega(a^*[H,a]) \geq 0\quad \forall a \in \cAloc.
\end{equation}
In the case of finite systems, this condition is equivalent to saying that the density matrix of $\omega$ is supported on the eigenspace of $H$ of minimal eigenvalue. To see this, we note that for finite systems, states are described by density matrices $\omega(a)=\tr[\rho a]$.
Choosing the operators $a$ in the above definition for a given eigenbasis of the Hamiltonian $e_i$ as $e_j e_i^*$, we see that
\begin{equation}
\rho_{ii}E_j-\rho_{ii}E_i\ge0
\end{equation}
where we denote matrix elements $\rho_{ij}$ in the energy eigenbasis and the corresponding eigenvalues by $E_i$.
This is fulfilled only if $\rho_{ii}=0$ for all $i$ corresponding to eigenvectors that are not in the ground-space sector.
By writing
\begin{equation}
\rho=\left(\begin{matrix}
    A&B\\
    B^*&C
\end{matrix}\right)
\end{equation}
where the blocks correspond to the ground-space and its orthogonal complement, we can see that $C$ has nonpositive diagonal entries, but is also positive semidefinite and thereby $0$.
By the Schur complement lemma we can then conclude that $B=0$ and thereby $\rho$ is supported on the ground space.

\begin{remark}[Relation to thermodynamic limit]
\label{rem:gsthermodynamiclimit}
Consider the case of a Hamiltonian on the infinite lattice $\sites = \ZZ^D$. It is immediate to verify that the limit of ground states of finite Hamiltonians $\tilde H_{\Lambda}$ of increasing size will automatically satisfy condition \eqref{eq:groundstate}. Actually this is even true for any choice of ``boundary condition'' $\Delta_{\Lambda^c}$ acting outside $\Lambda$. Indeed, one can easily show that if $\omega$ is such that
    \begin{equation}
    \omega(a) = \lim_{\Lambda \uparrow \ZZ^D} \<\Phi^{\Lambda} | a | \Phi^{\Lambda}\>
    \end{equation}
    where for each $\Lambda$, $\Phi^{\Lambda}$ is a ground state of the finite Hamiltonian $\tilde H_{\Lambda} + \Delta_{\Lambda^c}$, then $\omega$ will satisfy \eqref{eq:groundstate}.
Imposing such arbitrary boundary conditions while taking the limit may actually be needed in order to capture the different ground states in the thermodynamic limit.

To illustrate this, consider the following one-dimensional classical system with three states $\{0,1,2\}$ per site given by the Hamiltonian
\begin{align}
    h_{i,i+1}(s,s') = \left\{ \begin{array}{cc}
       0  & \text{ if } s = s' \in \{0,1\} \\
       2  & \text{ if } (s, s') \in \{(0,1), (1,0), (2,2), (2,1), (0,2), (1,2)\}  \\
       -1 & \text{ if } s = 2, s' = 0. \\
    \end{array}\right.
\end{align}
For the finite Hamiltonian $\tilde H_{\Lambda}$ with $\Lambda = \{-\ell, \dots, \ell\}$, the unique ground state is the sequence $(2,0,\ldots,0)$ with energy $-1$. \edit{This makes it also the ground state of the infinite system corresponding to the limit with open boundary condition.} Importantly, note that for the observable $O = \proj{0}$ at site $0$, we have that for any $\ell \geq 1$, its value in the ground state of $H_{\Lambda}$ is $1$. However, another valid ground state of the infinite chain is the all-1 state. For this ground state in the thermodynamic limit, the observable takes a value of $0$. \edit{It is realized by adding a boundary term $h_{\ell+1}(s)=-2\delta_{1,s}$ as it lowers the energy of the  all-1 configuration below the open boundary ground state.}

\edit{Beyond this simple but artificial example, there are many more classical and quantum models exhibiting an analogous effect that are of current interest in quantum many-body physics. Due to the more involved definitions and solutions of these models we refrain from giving them explicitly. Besides ubiquitous examples like the 1D Ising model which can be solved, in \cite{cha2017} the example of degenerate ground states of the quantum double model and an explicit construction of boundary conditions is given.}

\end{remark}

\paragraph{Thermal equilibrium states} For a finite system, the Gibbs state at inverse temperature $\beta=1/T \geq 0$ is given by the density matrix $\rho = e^{-\beta H}/\tr e^{-\beta H}$. An alternative characterization is via the so-called 
Kubo-Martin-Schwinger (KMS) conditions, which takes the form (for finite systems)
\begin{equation}
\label{eq:kmsfinite}
\omega(ba) = \omega(ae^{-\beta H} b e^{\beta H}) \quad \forall a,b \in \cAloc.
\end{equation}
It is easy to verify that the Gibbs state $\omega(a) = \tr(e^{-\beta H} a)/\tr e^{-\beta H}$ satisfies the equality conditions above, and that it is also the unique such state. To state the KMS-conditions for infinite systems, one first needs to introduce the time evolution operator
\begin{equation}
\alpha_t(a) = \lim_{\Lambda \uparrow \Gamma} e^{iH_{\Lambda} t} a e^{-iH_{\Lambda} t} \qquad \forall a \in \cAloc, t\in\RR
\end{equation}
where the limit is taken over a suitable sequence of finite subsets $\Lambda$.
For the local Hamiltonian interactions we consider here, $\alpha_t$ exists and is strongly continuous \cite[Theorem 6.2.4]{brbook}.
For a given time-evolution operator, there is a dense subset $\cA_\alpha\subset\cA$ such that $F(\alpha_t(a))$ is entire for all $a\in\cA_\alpha$ and bounded linear functionals $F$ \cite[Proposition 2.5.22]{brbook1}.
We say that a state $\omega:\cA\to \CC$ is a $\beta$-KMS state if \cite[Definition 5.3.1]{brbook}
\begin{equation}
    \label{eq:kms}
    \omega(ba) = \omega(a\alpha_{i\beta}(b)) \quad \forall a,b \in \cA_\alpha.
\end{equation}
In fact, it is sufficient to ensure this condition for any dense $\alpha$-invariant *-subalgebra of $\cA_\alpha$.
The set of $\beta$-KMS states will be denoted by $\KMS_{\beta}$:
\begin{equation}
\label{eq:Kbeta}
\KMS_{\beta} = \left\{\omega:\cA\to \CC : \omega(1) = 1, \; \omega(a^* a) \geq 0 \; \forall a \in \cAloc,\text{ and } \eqref{eq:kms} \right\}.
\end{equation}
For translation-invariant Hamiltonians on the lattice $\sites = \ZZ^D$ we will be mostly interested in translation-invariant KMS-states
\begin{equation}
    \label{eq:KbetaTI}
    \KMSTI_{\beta} = \left\{ \omega: \cA \to \CC : \omega \in \KMS_{\beta} \text{ and } \omega(a)=\omega(\tau_x(a))\; \forall a\in\cAloc,x\in\mathbb{Z}^D \right\}.
\end{equation}
The set $\KMS_{\beta}$ (and $\KMSTI_{\beta}$ for translation-invariant Hamiltonians) is nonempty for all $\beta \geq 0$. This can be shown by taking an appropriate limit of finite Gibbs states on a sequence of $\Lambda \uparrow \Gamma$. An important aspect of infinite systems (i.e., systems in the thermodynamic limit) is that the set of thermal equilibrium states $\KMS_{\beta}$ is not necessarily reduced to a singleton. The existence of many $\beta$-KMS states signals the presence of different thermodynamic phases.

\section{Convex relaxations via energy-entropy balance inequalities}

Our goal is to obtain a convex relaxation for the marginals of equilibrium states $\omega \in \KMS_{\beta}$.  A natural way to obtain a convex relaxation is to impose only a subset of the $\beta$-KMS equality conditions \eqref{eq:kms} (which corresponds to \eqref{eq:kmsfinite} for finite systems). The problem with the KMS linear equality constraints is that even if $a,b$ are local observables, the observable $a \alpha_{i\beta}(b)$ is in general not local (unless, say, the Hamiltonian is commuting).

\subsection{Energy-entropy balance inequalities}

The energy-entropy inequalities is a family of convex inequalities that characterize $\beta$-KMS states and that only involve local observables of $\omega$. They were first derived by Araki and Sewell in \cite{arakisewell}.

\begin{theorem}[{Energy-entropy balance (EEB) inequalities, see \cite[Theorem 5.3.15]{brbook}}]
\label{thm:EEB}
Let $H$ be a local Hamiltonian. Then $\omega:\cA\to \CC$ is a $\beta$-KMS state if, and only if, it satisfies\footnote{The function $x\log(x/y)$ is defined as equal to 0 if $x=0$ and $+\infty$ if $x>0$ and $y=0$.}
\label{eq:EEB}
\begin{subequations}
\begin{align}
& \omega([H,a]) = 0 \quad \forall a \in \cAloc,\label{eq:EEBcomm}\\
& \omega(a^* a) \log\left(\frac{\omega(a^* a)}{\omega(aa^*)}\right) \leq \beta \omega(a^* [H,a]) \quad \forall a \in \cAloc.\label{eq:EEBineq}
\end{align}
\end{subequations}
\end{theorem}
\begin{remark}
The cited theorem requires the conditions to hold for all elements in the domain of the generator $\delta$ of the time-evolution operator $\alpha_t$.
On $\cAloc$ this operator takes values $\delta(a)=i[H,a]$.
It is, however, defined on a larger domain given by the closure of this operator, i.e., the commutator formula on $\cAloc$ is a \textit{core} for $\delta$.
By the definition of the closure of an operator, $a\in\cA$ is said to be in the domain of $\delta$, i.e., $a\in D(\delta)$ if there exists a sequence $a_i\in\cAloc$ such that
$a_i\to a$ \textit{and} $\delta(a_i)\to \delta(a)$.
Due to the latter, it is sufficient to impose the EEB inequalities on $\cAloc$ only, despite the lack of continuity of $\delta$. 
\end{remark}

The convexity of the constraints \eqref{eq:EEBineq} in $\omega$ follow from convexity of the function $(x,y)\mapsto x\log(x/y)$. Furthermore if $H$ is a local Hamiltonian, and $a$ is a local operator, then the constraints \eqref{eq:EEB} only involve local observables of $\omega$. Note that for $\beta = +\infty$ (zero temperature), the constraint \eqref{eq:EEBineq} reduces to the inequality \eqref{eq:groundstate} which characterizes the ground states of $H$.

\begin{remark}[Relation to thermodynamic limit]
\label{rem:gibbsthermodynamiclimit}
Using conditions \eqref{eq:EEB} one can show that any limit of finite Gibbs states $\exp(-\beta H_{\Lambda})/\tr \exp(-\beta H_{\Lambda})$ as $\Lambda\uparrow \Gamma$ will be a valid $\beta$-KMS state. Actually the same is true for any sequence of finite Gibbs states of $H_{\Lambda} + \Delta_{\Lambda^c}$ where $\Delta_{\Lambda^c}$ is an arbitrary term acting only on the boundary of $\Lambda$.

Similarly to the example for ground-states in Remark \ref{rem:gsthermodynamiclimit}, we demonstrate how different thermal states can arise from different choices of boundary conditions. Consider the 2D Ising model at some temperature $0 < T < T_c$ (where $T_c$ is the critical temperature), and let for any $\ell \in \NN$, $\rho_{\ell,\beta}$ be the finite Gibbs state on a region $\{-\ell,\dots,\ell\}^2$ with open boundary condition. Then by the spin-flip symmetry of the Ising Hamiltonian, we know that $\tr(Z_x \rho_{\ell,\beta}) = 0$ for any site $x$ in the region $\{-\ell,\dots,\ell\}^2$. Thus this means that for any limit $\omega_{\beta}$ of the finite Gibbs state with open boundary condition, we will have $\omega_{\beta}(Z_x) = 0$. It is known however that for the 2D Ising model below the critical temperature, there are two (extremal) infinite-volume Gibbs states $\omega_{\beta}^+$ and $\omega_{\beta}^-$ such that $\omega_{\beta}^{-}(Z_x) < 0 < \omega_{\beta}^+(Z_x)$, which can be reached by taking limits of finite Hamiltonians with fixed boundary condition, i.e., aligned spins on the boundary \cite{fvbook}.
\end{remark}

\subsection{Convex relaxations}

The EEB inequalities above can be used to obtain rigorous lower and upper bounds on the value of any observable on thermal states via convex optimization.
Let $(\Lambda_{\ell})_{\ell \in \NN}$ be an increasing sequence of subsets of $\sites$ such that the support of $O$ is contained in $\Lambda_0$, and such that $\Lambda_{\ell} \uparrow \sites$ as $\ell\to\infty$. To simplify the presentation, we will further assume that $\overline{\Lambda_{\ell-1}} \subset \Lambda_{\ell}$ for all $\ell$, where, as before, $\overline{\Lambda} = \Lambda \cup \bdex \Lambda$ (see \eqref{eq:bdexLambda}).
For example, if $\sites = \ZZ^D$ and $H$ is a nearest-neighbor Hamiltonian and $O$ is supported on site $0$, one can take $\Lambda_{\ell} = \{-\ell,\ldots,\ell\}^D$.
For any $\ell$, let $\cA_{\ell} \subset \cA$ be the subalgebra of observables supported on $\Lambda_{\ell}$. Consider the optimization problem 
\label{eq:opt1}
\begin{subequations}
\begin{align}
\underset{\tomega:\cA_{\ell}\to \CC}{\text{min/max}} \quad & \tomega(O)\\
\text{s.t.} \quad & \tomega(1) = 1 \label{eq:opt1-norm}\\
            & \tomega(a^* a) \geq 0 \quad \forall a \in \cA_{\ell} \label{eq:opt1-pos}\\
            & \tomega([H,a]) = 0 \; \forall a \in \cA_{\ell-1} \label{eq:opt1-stationary}\\
            & \tomega(a^* a) \log\left(\frac{\tomega(a^* a)}{\tomega(aa^*)}\right) \leq \beta \tomega(a^* [H,a]) \quad \forall a \in \cA_{\ell-1}.\label{eq:opt1-EEB}
\end{align}
\end{subequations}
The feasible set of this convex optimization problem is an \emph{outer relaxation} for the set of $\Lambda_{\ell}$-marginals of thermal equilibrium states in $\cG_{\beta}$\edit{, i.e., the optimization variable takes values in the set of states on $\cA_\ell$}. This relaxation is obtained by imposing the EEB inequalities on the finite-dimensional subspace $\cA_{\ell-1}$, in addition to the normalization and positivity constraints \eqref{eq:opt1-norm}, \eqref{eq:opt1-pos}.
Besides the linear equations \eqref{eq:opt1-norm} and \eqref{eq:opt1-stationary}, the constraint \eqref{eq:opt1-pos} can be encoded as the positivity of a Hermitian matrix of size $|\Lambda_{\ell}| \times |\Lambda_{\ell}|$.  However, the last set of constraints \eqref{eq:opt1-EEB} form an infinite set of scalar inequalities, and it is unclear how to formulate them in a compact way. \edit{This is due to the lack of linearity in the inequality. Imposing the inequality for a basis of $\cA_{\ell-1}$is not sufficient to conclude it on all elements of $\cA_{\ell-1}$, which is an infinite set.} Clearly one can sample certain observables $a_i \in \cA_{\ell-1}$ ($i=1,\ldots,m$) and impose the $m$ scalar inequalities, however, this poses the question of how to choose these $a_i$'s.

\paragraph{Matrix EEB inequality} We introduce the following matrix version of the EEB inequality which allows us to impose the EEB inequality for all observables $a$ in a finite-dimensional subspace. The matrix EEB inequality makes use of the matrix relative entropy function \cite{fujii1989relative,fujii1990extension} which is defined as
\begin{equation}
D_{op}(A\|B) = A^{1/2} \log(A^{1/2} B^{-1} A^{1/2}) A^{1/2}
\end{equation}
for any two positive semidefinite matrices $A,B \psd 0$ such that $A \ll B$ (i.e., $\im\,A \subset \im\,B$). The function $D_{op}$ is the operator perspective \cite{effros2009matrix,ebadian2011perspectives} of the negative logarithm function (which is operator convex) and satisfies the following properties, see \cite{fujii1990extension}:
\begin{itemize}
\item Homogeneity: $D_{op}(\lambda A \| \lambda B) = \lambda D_{op}(A\|B)$ for $\lambda > 0$
\item Operator convexity\footnote{For homogeneous functions, convexity is equivalent to subadditivity}: $D_{op}(\sum_k A_k \| \sum_k B_k) \nsd \sum_k D_{op}(A_k \|B_k)$
\item Transformer inequality: For any rectangular matrix $K$, 
\begin{equation}
\label{eq:transformer}
D_{op}(K^* A K \| K^* B K) \nsd K^* D_{op}(A\|B) K.
\end{equation}
\end{itemize}
Furthermore, the function $D_{op}$ is closed, in the sense that $\{(A,B,T) : A\ll B \text{ and } D_{op}(A\|B) \nsd T\}$ is a closed convex set \cite[Theorem B.1]{sc}.
Convex optimization problems involving the matrix relative entropy can be solved using interior-point methods \cite{sc} or via semidefinite programming approximations \cite{logapprox}.

\begin{theorem}[Matrix EEB inequality]
\label{thm:matrixEEB}
Let $H$ be a local Hamiltonian and $\omega:\cA\to \CC$ be a $\beta$-KMS state. Let $a_1,\ldots,a_m \in \cAloc$ and define the $m\times m$ matrices
\begin{equation}
\label{eq:matrixEEB-ABC}
\begin{aligned}
A_{ij} &= \omega(a_i^* a_j)\\
B_{ij} &= \omega(a_j a_i^*)\\
C_{ij} &= \omega(a_i^* [H,a_j]).
\end{aligned}
\end{equation}
Then $A,B,C$ are Hermitian, $A,B$ are positive semidefinite with $A\ll B$ and $D_{op}(A\|B) \nsd \beta C$.
\end{theorem}
\begin{proof}
It is clear that $A,B$ are Hermitian and that they are positive semidefinite. That $A \ll B$ follows from the fact that for a KMS state $\omega(bb^*) = 0 \implies \omega(b^* b) = 0$ for any $b \in \cA$. (This can be seen e.g., from \eqref{eq:EEBineq}).
 To see that $C$ is Hermitian, note that since $\omega$ is a KMS state we have $\omega([H,b]) = 0$ for any $b \in \cAloc$ (see \eqref{eq:EEBcomm}). Then
\begin{equation}
C_{ij} = \omega(a_i^* [H,a_j]) = \omega(a_i^*(Ha_j - a_j H)) = \omega(a_i^* H a_j - H a_i^* a_j) = \omega([a_i^*,H] a_j) = \overline{\omega(a_j^* [H,a_i])} = \overline{C_{ji}}.
\end{equation}
We now prove the inequality $D_{op}(A\|B) \nsd \beta C$. We present the proof here in the finite-dimensional case. We treat the general case in Appendix \ref{sec:matrixEEBgeneral}. We thus assume that $\cA$ is a finite-dimensional matrix algebra. Let $L:\cA\to \cA$ be the commutator operator with the Hamiltonian $H$, i.e., $L(a) = [H,a]$.
The map $L$ is self-adjoint with respect to the Hilbert-Schmidt inner product. Observe that for any $a \in \cA$, $e^{-\beta L}(a) = e^{-\beta H} a e^{\beta H}$ (this can be easily seen e.g., by checking that $a(\beta) := e^{-\beta H} a e^{\beta H}$ is the solution of the ODE $a'(\beta) = -L(a(\beta))$ with $a(0) = a$.)
Let $L = \sum_{k} \lambda_k P_k$ be a spectral decomposition of $L$. By the KMS condition we have
\begin{equation}
B_{ij} = \omega(a_j a_i^*) = \omega(a_i^* e^{-\beta L}(a_j)) = \sum_{k} e^{-\beta \lambda_k} \omega(a_i^* P_k(a_j)) = \sum_{k} e^{-\beta \lambda_k} A^{(k)}_{ij}
\end{equation}
where we let $A^{(k)}_{ij} = \omega(a_i^* P_k(a_j))$. Since $\sum_k P_k = I$ and $\sum_k \lambda_k P_k = L$, we have
\begin{equation}
\sum_{k} A^{(k)} = A \quad \text{ and } \quad \sum_k \lambda_k A^{(k)} = C.
\end{equation}
Then we have
\begin{equation}
\begin{aligned}
D_{op}(A\|B) = D_{op}\left(\sum_k A^{(k)} \| \sum_k e^{-\beta \lambda_k} A^{(k)}\right)
             \nsd \sum_k D_{op}(A^{(k)} \| e^{-\beta \lambda_k} A^{(k)})
             = \sum_k \beta \lambda_k A^{(k)} = \beta C
\end{aligned}
\end{equation}
as desired.
\end{proof}

\paragraph{A convex relaxation using the matrix relative entropy} We now consider the modified convex optimization problem
\begin{subequations}
\label{eq:opt2}
\begin{align}
\underset{\tomega:\cA_{\ell}\to \CC}{\text{min/max}} \quad & \tomega(O)\\
\text{s.t.} \quad & \tomega(1) = 1 \label{eq:opt2-norm}\\
            & \tomega(a^* a) \geq 0 \quad \forall a \in \cA_{\ell} \label{eq:opt2-pos}\\
            & \tomega([H,a]) = 0 \; \forall a \in \cA_{\ell-1} \label{eq:opt2-comm}\\
            & D_{op}\left( \, [\tomega(a_i^* a_j)]_{ij} \, , \, [\tomega(a_j a_i^*)]_{ij} \, \right) \; \nsd \; \beta [\tomega(a_i^*[H,a_j])]_{ij}. \label{eq:opt2-EEB}
\end{align}
\end{subequations}
where in the last constraint, $\{a_i\}$ is a basis of $\cA_{\ell-1}$. One can show, using \eqref{eq:transformer}, that the infinite number of scalar constraints \eqref{eq:opt1-EEB} are implied by the finite positive semidefinite constraint \eqref{eq:opt2-EEB}. We prove this in the next proposition.

\begin{proposition}
Consider the optimization problem \eqref{eq:opt1} and assume $\{a_i\}$ ($i=1,\ldots,m$) is a basis of $\cA_{\ell-1}$. Then the constraints \eqref{eq:opt1-EEB} are all implied by the matrix inequality \eqref{eq:opt2-EEB}.
\end{proposition}
\begin{proof}
Let $a$ be an arbitrary element of $\cA_{\ell-1}$, and let $a = \sum_{i} x_i a_i$, with $x_i \in \CC$, be its decomposition in the basis $\{a_i\}$. Also let $A,B,C$ be the matrices defined in \eqref{eq:matrixEEB-ABC}, so that the inequality \eqref{eq:opt2-EEB} can be written as $D_{op}(A\|B) \nsd \beta C$. If we let $x=(x_1,\ldots,x_m) \in \CC^m$, then note that
\begin{equation}
\begin{aligned}
x^* A x &= \sum_{ij} \overline{x_i} x_j \omega(a_i^* a_j) = \omega(a^* a)\\
x^* B x &= \sum_{ij} \overline{x_i} x_j \omega(a_j a_i^*) = \omega(aa^*)\\
x^* C x &= \sum_{ij} \overline{x_i} x_j \omega(a_i^* [H, a_j]) = \omega(a^* [H,a]).
\end{aligned}
\end{equation}
It follows from \eqref{eq:transformer} that
\begin{equation}
\omega(a^* a) \log\left(\frac{\omega(a^* a)}{\omega(aa^*)}\right) = D_{op}(x^* A x \| x^* B x) \leq x^* D_{op}(A\|B) x \leq \beta x^* C x = \beta \omega(a^* [H,a])
\end{equation}
as desired.
\end{proof}

\subsection{Convergence}

Let $(\Pmin_{\ell})$ and $(\Pmax_{\ell})$ be respectively the minimization and maximization problems \eqref{eq:opt2}, and let $\pmin_{\ell} \leq \pmax_{\ell}$ be their optimal values. The next theorem shows the asymptotic convergence of  $\pmin_{\ell}$ and $\pmax_{\ell}$ to $\OminNTI_{\beta}$ and $\OmaxNTI_{\beta}$ respectively.

\begin{theorem}
\label{thm:main1}
Let $H$ be a local Hamiltonian, $O \in \cAloc$ a local observable, and $\beta \in [0,+\infty]$. 
Define
\begin{equation}
[\OminNTI_{\beta}, \OmaxNTI_{\beta}] = \{ \omega(O) : \omega \in \KMS_{\beta} \}.
\end{equation}
Then $\pmin_{\ell} \uparrow \OminNTI_{\beta}$ and $\pmax_{\ell} \downarrow \OmaxNTI_{\beta}$ as $\ell \to \infty$.
\end{theorem}
\begin{proof}
Note that $\KMS_{\beta}$ is convex and weak-* compact \cite[Theorem 5.3.30]{brbook} and so $\{ \omega(O) : \omega \in \KMS_{\beta}\}$ is a closed interval.

If $\omega \in \KMS_{\beta}$, then for any $\ell$, its restriction to $\cA_{\ell}$ is feasible for the optimization problems \eqref{eq:opt2}. Thus this means that $\pmin_{\ell} \leq \OminNTI_{\beta} \leq \OmaxNTI_{\beta} \leq \pmax_{\ell}$ for any $\ell$. Furthermore, observe that $\pmin_{\ell}$ is monotonic nondecreasing, and $\pmax_{\ell}$ is monotone nonincreasing. Indeed assume $\tomega_{\ell}:\cA_{\ell}\to \CC$ is feasible for \eqref{eq:opt2} and consider its restriction to $\cA_k$ for $k < \ell$. We want to show that it is feasible for the level $k$ of the relaxation \eqref{eq:opt2}. It is immediate that the restriction satisfies the constraints \eqref{eq:opt2-norm}-\eqref{eq:opt2-comm} at the level $k$ of the relaxation. To prove \eqref{eq:opt2-EEB} we use the transformer inequality \eqref{eq:transformer}. Let $\{b_s\}$ be a basis of $\cA_{k-1}$ and note that we can express each $b_s$ as a linear combination of the $a_i \in \cA_{\ell-1}$, i.e., $b_s = \sum_{i} K_{is} a_i$ for some $K_{is} \in \CC$. As such we have
\begin{equation}
\tomega(b_s^* b_t) = \sum_{ij} \overline{K_{is}} K_{jt} \tomega(a_i^* a_j)
\end{equation}
so that we can write $[\tomega(b_s^* b_t)]_{s,t} = K^* [\tomega(a_i^* a_j)]_{i,j} K$
and similarly for $[\tomega(b_t b_s^*)]_{s,t}$. Hence we get
\begin{equation}
\begin{aligned}
D_{op}([\tomega(b_s^* b_t)]_{s,t} \; , \; [\tomega(b_t b_s^*)]_{s,t}) &= 
D_{op}(K^* [\tomega(a_i^* a_j)]_{i,j}  K \; , K^* [\tomega(a_j a_i^*)]_{i,j} K)\\
&\nsd K^* D_{op}([\tomega(a_i^* a_j)]_{i,j} \; , \; [\tomega(a_j a_i^*)]_{i,j}) K\\
&\nsd \beta K^* [\tomega(a_i^*[H,a_j])]_{i,j} K\\
&= \beta [\tomega(b_s^*[H,b_t])]_{s,t},
\end{aligned}
\end{equation}
as desired.

Since $(\pmin_{\ell})$ and $(\pmax_{\ell})$ are monotonic and bounded
they must have limits $\pmin_{\infty}:=\lim \pmin_{\ell}$ and $\pmax_{\infty}:=\lim \pmax_{\ell}$ respectively. We want to show that $\pmin_{\infty}=\OminNTI_{\beta}$ and $\pmax_{\infty}=\OmaxNTI_{\beta}$.

For each $\ell$, let $\tomega_{\ell}:\cA_{\ell} \to \CC$ be a feasible solution to \eqref{eq:opt2} such that $\pmin_{\ell} \leq \tomega_{\ell}(O)\leq \pmin_{\ell} + 1/\ell$ (if the minimum value of the optimization problem is attained, we simply take $\tomega_{\ell}$ the optimal solution which satisfies $\tomega_{\ell}(O) = \pmin_{\ell}$). We can extend $\tomega_{\ell}$ to a state on the full $C^*$-algebra $\cA$ \cite[Proposition 2.3.24]{brbook1}. We now have a sequence of states $\tomega_{\ell}$ on $\cA$. Since the set of states on $\cA$ is weak-* compact \cite[Theorem 2.3.15]{brbook1}, we can extract from $\tomega_{\ell}$ a convergent subnet $\tomega'_{i}$ defined on the directed set $\mathcal{I}$ via a monotone final function $h:\mathcal I\to \mathbb N$, $\tomega'_i:=\tomega_{h(i)}$. We denote the limit of the subnet by $\tomega:=\lim_i \tomega'_i$, which is a state on $\cA$. 

We now show that $\tomega$ is a $\beta$-KMS state by showing that the conditions \eqref{eq:EEB} hold.
Let $a \in \cAloc$. Since $a$ is a local observable we have for some $\ell'$ that $a \in \cA_\ell$ for all $\ell\ge\ell'$. Hence this implies that
\begin{equation}
\begin{aligned}\label{eq:EEBlimitCond}
& \tomega_{\ell}([H,a]) = 0\\
& \tomega_{\ell}(a^* a) \log\left(\frac{\tomega_{\ell}(a^* a)}{\tomega_{\ell}(aa^*)}\right) \leq \beta\tomega_{\ell}(a^*[H,a])
\end{aligned}
\end{equation}
for all $\ell\ge\ell'$. Since $h$ is final and monotone, there exists $i'\in\mathcal I$, such that $h(i)\ge\ell'$ $\forall i\ge i'$.
Let $S$ be the complement of the set of states on which~\eqref{eq:EEBlimitCond} holds, which is open by weak-* continuity of $\omega\mapsto\omega(O)$ for any $O\in\cA$.
Assume $\tomega\in S$. Then there exists $k\in\mathcal{I}$ such that $\forall i\ge k$ we have $\tomega_i'\in S$. However, since $\mathcal I$ is a directed set, there exists $i\in\mathcal I$ with $i\ge i'$ and $i\ge k$ and so $\tomega_{h(i)}\not\in S$ and $\tomega_{h(i)}\in S$, which is a contradiction.
Therefore we conclude that $\tomega$ fulfills \eqref{eq:EEBlimitCond} or likewise \eqref{eq:EEB} and is a $\beta$-KMS state.

Finally note that
\begin{equation}
\tomega(O) = \lim_{i} \tomega'_i(O)=\lim_{i} \tomega_{h(i)}(O) \leq \lim_{i} \pmin_{h(i)}+1/h(i) = \pmin_{\infty}.
\end{equation}
Here, $\pmin_{h(i)}$ and $1/h(i)$ are subnets of the convergent sequences $\pmin_\ell$ and $1/\ell$ and thereby converge.
Since $\pmin_{\infty} \leq \OminNTI_{\beta} \leq \tomega(O)$ we necessarily get equality $\pmin_{\infty} = \OminNTI_{\beta}$. The same argument can be used to show that $\pmax_{\infty} = \OmaxNTI_{\beta}$.
\end{proof}

\paragraph{Translation-invariant Hamiltonians} We now prove \mainref{Theorem~\ref{thm:mainTI}}{the main theorem, which we restate here}.
For translation-invariant Hamiltonians, one can adapt the hierarchy \eqref{eq:opt2} so that the resulting sequences $(\pminTI_{\ell})$ and $(\pmaxTI_{\ell})$ converge to $\Omin$ and $\Omax$ respectively. It suffices to add the following translation-invariance linear constraint in the convex optimization problem:
\begin{equation}
    \label{eq:opt2-TI-new}
    \tomega(\tau_x(a)) = \tomega(a) \quad \forall a \in \cA_{\ell} \; \forall x \in \ZZ^D \text{ s.t. } \tau_x(a) \in \cA_{\ell}.
\end{equation}
To prove convergence, one simply needs to show (using the same notations as in the proof of Theorem \ref{thm:main1}) that $\tomega = \text{weak*-lim} \, \tomega_{n(\ell)}$ is translation-invariant. To do this, let $a \in \cAloc$ and $x \in \ZZ^D$. Since $a$ is a local operator, we know that $a,\tau_x(a) \in \cA_{n(\ell)}$ for all large enough $\ell$. This means that
\begin{equation}
\tomega(\tau_x(a)) = \lim_{\ell} \tomega_{n(\ell)}(\tau_x(a)) = \lim_{\ell} \tomega_{n(\ell)}(a) = \tomega(a)
\end{equation}
where the middle equality follows from the translation-invariant constraint \eqref{eq:opt2-TI-new} imposed in the optimization problem. This is true for all $a \in \cAloc$ and $x \in \ZZ^D$, and so $\tomega$ is translation-invariant.

\subsection{Decidability}
It has been shown in \cite{bausch2021uncomputability} that computing the phase diagram of a 2D nearest neighbour translation-invariant Hamiltonian in the case $\beta=\infty$ is undecidable.
This is shown for both a definition of phase in terms of gapped and gapless Hamiltonians as well as in terms of an order parameter.
It is interesting to compare the latter with our result.
More precisely, the construction of the authors gives a Hamiltonian (described by a nearest neighbour interaction $h(n)$) for every input $n$ to a universal Turing machine and a local observable $O$ such that
\begin{itemize}
\item If the Turing machine halts on input $n$ then for any sequence of ground states $\ket{\psi_L}$ of the Hamiltonian on an $L\times L$ square with open boundary conditions $\lim_{L\to\infty} \bra{\psi_L}O\ket{\psi_L}=0$.
\item If the Turing machine does not halt on input $n$, then for any sequence of ground states $\ket{\psi_L}$ of the Hamiltonian on an $L\times L$ square with open boundary conditions $\lim_{L\to\infty} \bra{\psi_L}O\ket{\psi_L}=1$.
\end{itemize}
\edit{It is interesting to compare this to our algorithm by asking what the output of our algorithm would be given the construction above as its input.
We leave a full resolution of this question to future work.
However, in light of the proven undecidability we can rule out that there exists a constant $\eps<1/2$ such that the interval provided by our algorithm is below $\eps$ for all Hamiltonians arising from halting instances and above $1-\eps$ for all non-halting ones.\\ \indent
We can however say the following to partially explain how the proof in \cite{bausch2021uncomputability} would not extend to arbitrary boundary conditions.}
The limit values in this construction crucially depend on the choice of boundary conditions.
While the boundary condition is not mentioned explicitly as it is always kept open, in fact a technique from \cite{gottesman2010quantum} is used to encode a boundary term corresponding to an energy shift of a Hamiltonian into the translation-invariant interaction term.
Before this shift the construction yields a nonnegative ground-state energy in the nonhalting case and a negative one in the halting case.
Adding this energy shift and adding a subspace $\{\ket0\}$ with a trivial product state with zero energy to the Hamiltonian, the ground state in the nonhalting case will be given by this product state.
The observable $O=\ketbra{0}{0}$ is then given by a projector onto this subspace.

The definition of ground-states in our work is different and in fact the set of equilibrium states $\KMS_{\infty}$ includes the limits of ground-states of Hamiltonians for \textit{any} boundary condition, see Remark \ref{rem:gsthermodynamiclimit}. 
In this version, as opposed to the one above, our hierarchies do prove the decidability of the equilibrium observable problem as given in  \mainref{Definition~\ref{def:eop}}{the following definition}.
\edit{This means that the problem solution is defined over the set of all boundary conditions, not that the problem can be decided for every choice of boundary condition.}

\begin{theorem}
    The Equilibrium Observable Problem is decidable.
\end{theorem}
\begin{proof}
    Let us start with the case $\beta=\infty$.
    The algorithm consists of iterating over $\ell=2,3,\ldots$ and alternating between the minimization and maximization of the convex relaxations~\eqref{eq:opt2} with the translation-invariant constraint \eqref{eq:opt2-TI-new}.
    Note that for the ground-state Eq.~\eqref{eq:opt2-EEB} reduces to $0\nsd \; \beta [\tomega(a_i^*[H,a_j])]_{ij}$, such that the problem becomes an SDP in rational coefficients that can be solved exactly.
    More specifically, we use quantifier elimination \cite{anai1998,basu2006} to decide whether the hierarchy augmented by $\tomega(O)> a$ or augmented by $\tomega(O)< a$ is feasible.
    The algorithm stops as soon as one of them becomes infeasible and outputs NO or YES respectively.
    Due to Theorem~\ref{thm:main} this is guaranteed to happen at sufficiently large $\ell$.

    In the case $\beta<\infty$, we face additional difficulties due to the matrix logarithm, as our convex relaxation is no longer an SDP.
    We will use SDP approximations of the logarithm, however, those have convergence guarantees that are only uniform on intervals that are bounded away from 0.
    A way to deal with this issue is to strengthen our convex relaxation by an \textit{a priori} lower bound on the moment matrices derived from an analysis of the Hamiltonian.
    This is proven in the following Lemma by adapting a result from \cite{bakshi2023} to the infinite system moment problem setting.
    \begin{lemma}\label{lem:momentLower}
    For a given local translation-invariant Hamiltonian $H$ on $\mathbb{Z}^D$ and temperature $\beta$, there exists a KMS-state $\omega\in\KMS_\beta$ such that the following is true: for any finite set $\Lambda\subset\ZZ^D$, there exists an explicit constant $C_\Lambda >0$ depending on $\beta$, the locality and norm of the Hamiltonian, and $\Lambda$, such that the marginal $\rho_\Lambda\in\cA_\Lambda$ of $\omega$, defined by
    \begin{equation}
    \tr[\rho_{\Lambda} a]=\omega(a)\qquad\forall a\in\cA_\Lambda
    \end{equation}
    is lower bounded
    \begin{equation}
    \rho_\Lambda\psd C_\Lambda.
    \end{equation}
    \end{lemma}
    \begin{proof}
    Without loss of generality we assume that the Hamiltonian is defined on qubits since general qudit systems can be embedded by encoding each qudit in a Hilbert space $\CC^{2^{m^D}}$ such that $d\le 2^{m^D}$, corresponding to a hypercube of spins. This results in a Hamiltonian with higher but still finite range. 
    
    In \cite[Corollary 2.14]{bakshi2023} it is shown that for a given qubit Hamiltonian $H_V$ on a finite set $V$ and for a fixed $\Lambda \subset V$ there exist explicit constants $c_{\Lambda}, c'_{\Lambda}$ such that the following is true: for any Hermitian observable $X=\sum_i x_i a_i$ supported on $\Lambda \subset V$ (here,  $x_i\in \RR$ and $a_i$ is the orthonormal Pauli basis supported on $\Lambda$) we have 
    \begin{equation}
    \tr[X^2\tilde{\rho}_V]\ge \exp(-c_{\Lambda}-\beta c'_{\Lambda})\max x_i^2,
    \end{equation}
    where $\tilde{\rho}_V = e^{-\beta H_V}/\tr e^{-\beta H_V}$ is the Gibbs state at inverse temperature $\beta$ for the given finite Hamiltonian $H_V$.
    Crucially, the constants $c_{\Lambda},c'_{\Lambda}$ do not depend on $|V|$, but only on the locality of the Hamiltonian and the degree of its so-called dual interaction graph, which are fixed.\footnote{The dual interaction graph has nodes for each Hamiltonian term and edges between terms that overlap}
    If we take the weak*-limit of Gibbs states $\tilde{\rho}_V$ on increasing system sizes $V\uparrow \ZZ^D$ (taking an appropriate subsequence to ensure convergence, see \cite[Proposition 6.2.15]{brbook}) we get a KMS-state $\omega$ which satisfies
    \begin{equation}
    \omega\left(X^2\right)\ge \exp(-c_{\Lambda}-\beta c'_{\Lambda})\max_i x_i^2
    \end{equation}
    for all $X \in \cA_{\Lambda}$, where the $x_i$ are the coefficients of $X$ in the Pauli basis of $\cA_{\Lambda}$.

    Now let $\rho_{\Lambda}$ be the marginal of $\omega$ on $\Lambda$, and let $X$ be the projector onto a one-dimensional eigenspace of $\rho_\Lambda$ with lowest eigenvalue.
    Note that the $x_i$ are real as $X$ is a Hermitian operator and that
    \begin{equation}
    \max_i x_i^2\ge \frac{1}{4^{|\Lambda|}}\|(x_i)_i\|_2^2=\frac{1}{4^{|\Lambda|}}\|X\|_F^2
    =\frac{1}{4^{|\Lambda|}}.
    \end{equation}
    Thereby, we conclude
    \begin{equation}
    \rho_\Lambda\psd \tr\left[\rho_\Lambda X^2\right]=\omega\left(X^2\right)\ge\frac{1}{4^{|\Lambda|}}\exp(-c_{\Lambda}-\beta c'_{\Lambda}) =: C_\Lambda.
    \end{equation}
    \end{proof}
    We now choose the basis of $\cA_\Lambda$ in our convex relaxation~\eqref{eq:opt2} to be orthonormal.
    We see that the above Lemma proves a lower bound on the lowest eigenvalue of the moment matrix $\omega$ that corresponds to the KMS-state given in the Lemma:
\begin{equation}\label{eq:strong-pos}
    \min_{\|x\|_2^2=1}
    \bra{x}[\tomega(a_i^*a_j)]_{ij}\ket{x}=\min_{\|x\|_2^2=1}\sum_{ij} \overline{x_i}x_j\omega(a_i^*a_j)=\tr[\rho X^*X]\ge C_{\Lambda_\ell}\tr[X^*X]=C_{\Lambda_\ell}
    \end{equation}
    and similarly for $[\tomega(a_j a_i^*)]_{ij}$.
    We use this to strengthen our convex relaxation~\eqref{eq:opt2} by adding the contraints
    \begin{equation}
    \label{eq:lbmomconstraint}
    \begin{aligned}
    [\tomega(a_i^*a_j)]_{ij}&\psd C_{\Lambda_\ell}'>0,\\
    [\tomega(a_j a_i^*)]_{ij}&\psd C_{\Lambda_\ell}'>0
    \end{aligned}
    \end{equation}
    where $C'_{\Lambda_{\ell}}$ is a rational lower bound to $C_{\Lambda_{\ell}}$.
    Note that Lemma~\ref{lem:momentLower} need not hold for KMS states that arise from different boundary conditions as those might come with different constants $c$, $c'$.
    In the promise problem setting this is not relevant though as the relaxation still contains at least the open boundary Gibbs state and thereby the respective SDP problem is always feasible.
    Therefore, the optimal values of the modified problems are two sequences $\pmin_{\ell}\le\pminp_{\ell}\le\pmaxp_{\ell}\le\pmax_{\ell}$.
    Note that monotonicity and thereby convergence are no longer guaranteed for $\pminp_\ell, \pmaxp_\ell$.
    However, we only need that the optimal values exceed the decision threshold, which is guaranteed by convergence of $\pmin_{\ell}, \pmax_{\ell}$ and the above inequality.
    
    After this modification, since the moment matrix is bounded away from zero and upper bounded by 1, the eigenvalues in the argument of the logarithm for the matrix relative entropy are contained in an interval $[a,1/a]$ for some $a>0$. We now introduce a semidefinite approximation of $D_{op}$ from \cite{logapprox}. For any integer $m \geq 1$ let $h_m(x) = 2^m \left(x^{1/2^m}-1\right)$ which satisfy $\log(x) \leq h_{m+1}(x) \leq h_m(x)$ and note that $h_m \to \log$ uniformly on any compact interval in $(0,\infty)$ since $0 \leq h_m(x) - \log(x) \leq 2^{-m}[(x-1)^2+(x^{-1}-1)^2]$ for all $x > 0$. Consider the \emph{noncommutative perspective} \cite{effros2009matrix,ebadian2011perspectives} of $-h_m$ which we denote by $D_{op}^{[m]}$:
    \begin{equation}
    D_{op}^{[m]}(X,Y) := P_{-h_m}(X,Y) = -X^{1/2} h_m(X^{-1/2} Y X^{-1/2}) X^{1/2}.
    \end{equation}
    Then $D_{op}^{[m]}$ is jointly convex in $(X,Y)$, has a finite semidefinite programming formulation \cite{logapprox}, and satisfies $D_{op}^{[m]}(X\|Y) \nsd D_{op}(X\|Y)$ since $-h_m \leq -\log$.
    By changing the function $D_{op}$ in the entropy constraint \eqref{eq:opt2-EEB} by $D_{op}^{[m]}$ we get a semidefinite program $(\Pminp_{\ell,m})$ with rational coefficients  whose optimal value satisfies $\pminp_{\ell,m}\le\pminp_\ell$ and analogously for the upper bounds.
    
    Furthermore using the constraint \eqref{eq:lbmomconstraint}, the uniform convergence of $h_m$ to $\log$, and the continuity of $D_{op}$ on positive definite matrices, we know that $\pminp_{\ell,m}\uparrow \pminp_{\ell}$ and similarly $\pmaxp_{\ell,m} \downarrow \pmaxp_{\ell}$ as $m\uparrow \infty$.
    By choosing an exhaustive sequence $(\ell_i,m_i)$ of $\mathbb{N}^+\times\mathbb{N}^+$ and deciding the two feasibility problems as in the ground-state case we obtain an algorithm that decides the problem in finite time.
    
\end{proof}

\subsection{Commuting Hamiltonians and the high-temperature regime}

In this section we focus on commuting Hamiltonians $H$ which satisfy $[h_X,h_Y] = 0$ for any $X,Y \subset \sites$. For concreteness, we choose to focus on $H$ translation-invariant on $\Gamma = \ZZ^D$. For such Hamiltonians one can define a convex relaxation for the set of thermal equilibrium states directly from the KMS equality condition. Indeed, in this case if $b \in \cA_{\Lambda} \subset \cAloc$ where $\Lambda$ is finite, then
\begin{equation}
\label{eq:alphatb}
\alpha_{t}(b) = \lim_{\Lambda' \uparrow \ZZ^D} e^{itH_{\Lambda'}} b e^{-itH_{\Lambda'}} = e^{it\tilde{H}_{\Lambda}} b e^{-it \tilde{H}_{\Lambda}}
\end{equation}
where $\tilde{H}_{\Lambda}$ is as defined in \eqref{eq:HtildeLambda}.
Clearly, in this case $\alpha_t(b)$ can be extended to an entire function for all $b\in\cAloc$. 
As such, the $\beta$-KMS condition for a state $\omega$ can be written as
\begin{equation}
\label{eq:localKMScomm}
\omega(ba) = \omega(ae^{-\beta \tilde{H}_{\Lambda}} b e^{\beta \tilde{H}_{\Lambda}}) \qquad \forall a \in \cA_{\bar \Lambda}, \; b \in \cA_{\Lambda}
\end{equation}
for all $\Lambda \subset \ZZ^D$ finite. (Recall that $\overline \Lambda = \Lambda \cup \bdex \Lambda$ where $\bdex \Lambda$ is the external boundary of $\Lambda$.) A consequence of these KMS conditions\footnote{This is because $\alpha_t(h_X) = h_X$, since the Hamiltonian is commuting.} is that
\begin{equation}
\label{eq:localKMScomm2}
\omega([a,h_X]) = 0 \;\; \forall a \in \cA_{\bar \Lambda}, \; \forall X \subset \bar \Lambda
\end{equation}
for all $\Lambda \subset \ZZ^D$ finite. Given a fixed $\Lambda \subset \ZZ^D$ finite, let\footnote{Note that \eqref{eq:localKMScomm} for a fixed $\Lambda$ only implies \eqref{eq:localKMScomm2} for $X \subset \Lambda$ (and not $\bar \Lambda$). That's why we impose \eqref{eq:localKMScomm2} explicitly.}
\begin{equation}
\label{eq:GLambdacomm}
\begin{aligned}
\KMS_{\Lambda,\beta}^{\comm} = \Bigl\{ \omega : \cA_{\bar \Lambda} \to \CC \text{ s.t. }& \omega(1) = 1, \omega(a^* a) \geq 0 \;\; \forall a \in \cA_{\bar \Lambda}, \\
&\omega(a)=\omega(\tau_x(a))\; \forall a\in\cA_{\bar \Lambda},x\in \ZZ^D \text{ s.t. } \tau_x(a) \in \cA_{\bar \Lambda}\\
&\text{ and } \eqref{eq:localKMScomm} \text{ and } \eqref{eq:localKMScomm2} \Bigr\}.
\end{aligned}
\end{equation}
Since \eqref{eq:localKMScomm} is bilinear in $(a,b)$ it only needs to be imposed on a basis $\{a_i\}$ of $\cA_{\bar \Lambda}$ and $\{b_j\}$ of $\cA_{\Lambda}$. As such, the convex set \eqref{eq:GLambdacomm} is the feasible set of a semidefinite program. It is an outer relaxation for the set of $\bar \Lambda$-marginals of $\beta$-KMS states.

Given an observable $O \in \cA_{\Lambda}$, one can consider the convex optimization problems
\begin{equation}
\label{eq:optcomm}
\min/\max \; \{ \tomega(O) \text{ s.t. } \tomega \in \cG_{\Lambda,\beta}^{\comm} \}.
\end{equation}
The solution of this minimization (resp. maximization) problem is a \emph{lower bound} on $\Omin_{\beta}$ (resp. \emph{upper bound} on $\Omax_{\beta}$).

The asymptotic convergence of the hierarchy \eqref{eq:optcomm} for a sequence of increasing finite lattices $\Lambda_{\ell}$ to $\Omin_{\beta}$ and $\Omax_{\beta}$ can be proved in exactly the same way as in Theorem \ref{thm:main1}. For this hierarchy however we are also able to give a quantitative rate of convergence in the high temperature regime.

\begin{theorem}
\label{thm:maincomm1}
Assume that the Hamiltonian is translation-invariant, commuting and that the lattice dimension $D \leq 2$. For $D=1$ and $\beta_1=\infty$ or for $D=2$ and some $\beta_1 > 0$, we have for all $0 \leq \beta < \beta_1$, and for any local observable $O$ the following:
\begin{itemize}
    \item $\Omin_{\beta} = \Omax_{\beta} = \OTI_{\beta}$
    \item If we let $\pmin_{\ell}$ and $\pmax_{\ell}$ be the optimal values of the convex relaxations \eqref{eq:optcomm} with $\Lambda = \Lambda_{\ell} = \{-\ell,\ldots,\ell\}^D$, then $\pmin_{\ell} \leq \OTI_{\beta} \leq \pmax_{\ell}$, and furthermore
    \begin{equation}
\max \left\{ \OTI_{\beta} - \pmin_{\ell} , \pmax_{\ell} - \OTI_{\beta} \right\} \leq c_1 \|O\| e^{-c_2 \ell}.
\end{equation}
for some constants $c_1,c_2 > 0$ depending on the interaction, temperature and the support of $O$.
\end{itemize}
\end{theorem}
\begin{proof}
We start with an important lemma showing that states in $\KMS_{\Lambda,\beta}^{\comm}$ can be expressed as a product of the local (i.e., finite) Gibbs state on $\Lambda$ and a boundary term on $\bdex \Lambda$, after a suitable perturbation that eliminates the surface interaction terms. The statement is closely related to \cite[Prop. 6.2.17]{brbook}. The difference is that the state $\omega$ in the lemma below is not assumed to be a ``full'' KMS state, but just a state in $\KMS_{\Lambda,\beta}^{\comm}$.

\begin{lemma}\label{lem:commutingPerturbed}
Assume $H$ is a commuting Hamiltonian. Let $\Lambda \subset \ZZ^D$ be a finite set, and let $\omega \in \KMS_{\Lambda,\beta}^{\comm}$.
Let 
\begin{equation}
W = \sum_{\substack{X\cap \Lambda \neq \emptyset\\ X \cap \Lambda^c \neq \emptyset}} h_X \in \cA_{\bar \Lambda}
\end{equation}
be the surface interaction term.
Let
\begin{equation}
\varphi^{\Lambda}(x) = \frac{\tr (e^{-\beta H_{\Lambda}}x)}{\tr e^{-\beta H_{\Lambda}}}.
\end{equation}
There exists a state $\hat\omega:\cA_{\bdex\Lambda}\to\CC$ on the boundary such that using the product state $\omega^W:\cA_{\bar\Lambda}\to\CC$ defined by
\begin{equation}
\omega^W(xc)=\varphi^\Lambda(x)\hat\omega(c)
\end{equation}
we can express $\omega$ as a perturbation of the local Gibbs state with a boundary condition
\begin{equation}
\omega(O)=\frac{\omega^W(Oe^{-\beta W})}{\omega^W(e^{-\beta W})}.
\end{equation}
\end{lemma}
\begin{proof}
We define 
\begin{equation}
\omega^{W}(a) = \frac{\omega(e^{\beta W} a)}{\omega(e^{\beta W})} = \frac{\omega(a e^{\beta W})}{\omega(e^{\beta W})},
\end{equation}
where the equality follows using the $\beta$-KMS condition and the fact that $H$ is commuting. We first prove that it has the claimed form, i.e., is product with the factor on $\Lambda$ being $\varphi^\Lambda$.
Let $a,b \in \cA_{\Lambda}$ and $c \in \cA_{\bdex \Lambda}$. Then condition \eqref{eq:localKMScomm} tells us that
\begin{equation}
\label{eq:localKMSbac}
\begin{aligned}
\omega^{W}(bac) &= \frac{1}{\omega(e^{\beta W})} \omega(bac e^{\beta W})\\
                &= \frac{1}{\omega(e^{\beta W})} \omega(ac e^{\beta W} e^{-\beta \tilde{H}_{\Lambda}} b e^{\beta \tilde{H}_{\Lambda}})\\
                &= \frac{1}{\omega(e^{\beta W})} \omega(ac e^{-\beta H_{\Lambda}} b e^{\beta H_{\Lambda}} e^{\beta W})\\
                &= \omega^{W}(ac e^{-\beta H_{\Lambda}} b e^{\beta H_{\Lambda}})\\
                &= \omega^{W}(a e^{-\beta H_{\Lambda}} b e^{\beta H_{\Lambda}}c)
\end{aligned}
\end{equation}
where in the last line we used the fact that $c \in \cA_{\bdex \Lambda}$ and so it commutes with $e^{-\beta H_{\Lambda}} b e^{\beta H_{\Lambda}} \in \cA_{\Lambda}$.

From \eqref{eq:localKMSbac} we get that for any fixed $c \in \cA_{\bdex \Lambda}$, the state $x \in \cA_{\Lambda} \mapsto \omega^{W}(xc)/\omega^{W}(c)$ satisfies the KMS condition for the finite-dimensional Hamiltonian $H_{\Lambda}$. Since this is a finite system, this means that necessarily for any $x \in \cA_{\Lambda}$,
\begin{equation}
\frac{\omega^{W}(xc)}{\omega^{W}(c)} = \frac{\tr(e^{-H_{\Lambda}} x)}{\tr e^{-H_{\Lambda}}}.
\end{equation}
This proves that $\omega^W$ is product and has the marginal on $\Lambda$ as claimed.
We conclude by computing by
\begin{equation}
\omega(O)=\frac{\omega(O)}{\omega(e^{\beta W})}\frac{\omega(e^{\beta W})}{\omega(\id)}=\frac{\omega^W(Oe^{-\beta W})}{\omega^W(e^{-\beta W})},
\end{equation}
using the definition of $\omega^W$ twice.
This proves Lemma \ref{lem:commutingPerturbed}.
\end{proof}

We are now ready to prove exponential convergence of the convex optimization hierarchy to $\<O\>_{\beta}$ for $0\leq \beta < \beta_1$. We start with the 1D case. Our main tool will be the decay of correlations and local indistinguishability of Gibbs states for one-dimensional chains \cite{bluhm2022}. Without loss of generality let us assume a 2-local interaction; the general case follows from blocking the sites first. Let $\omega \in \KMS_{\Lambda_{\ell},\beta}^{\comm}$ and consider the perturbation $\omega^{W}$ from Lemma \ref{lem:commutingPerturbed}. The surface interaction term $W$ has support on the sites $-(\ell+1),-\ell,\ell,\ell+1$ and $\|\exp(\pm\beta W)\|\le\exp(2\beta\|h\|)$ is bounded by a constant.
We first apply the decay of correlations to bound the effect of the perturbation.
From Lemma~\ref{lem:commutingPerturbed} we have
\begin{equation}
\omega(O)=\frac{\omega^W(Oe^{-\beta W})}{\omega^W(e^{-\beta W})}.
\end{equation}
Let us introduce the operator $Y=(id_{\Lambda} \otimes \hat\omega_{\bdex \Lambda})(e^{-\beta W})\in\cA_\Lambda$, which satisfies $\|Y\|\le\exp(2\beta\|h\|)$.
We further decompose $Y=Y_LY_R$ where $Y_L$, $Y_R$ are supported on $-\ell$ and $\ell$ respectively.
We can write
\begin{equation}
\omega(O)=\frac{\varphi^\Lambda(OY)}{\varphi^\Lambda(Y)}
\end{equation}
and thereby
\begin{align}
|\omega(O)-\varphi^\Lambda(O)|&\le\frac{|\varphi^\Lambda(OY)-\varphi^\Lambda(OY_L)\varphi^\Lambda(Y_R)|}{\varphi^\Lambda(Y)}\\
&\quad+\frac{|\varphi^\Lambda(OY_L)\varphi^\Lambda(Y_R)-\varphi^\Lambda(Y_L)\varphi^\Lambda(Y_R)\varphi^\Lambda(O)|}{\varphi^\Lambda(Y)}\\
&\quad+\frac{|\varphi^\Lambda(O)\varphi^\Lambda(Y_L)\varphi^\Lambda(Y_R)-\varphi^\Lambda(O)\varphi^\Lambda(Y)|}{\varphi^\Lambda(Y)}\\
&\le \|O\| K'\exp(-\alpha \ell)
\end{align}
where to each term we can apply the uniform clustering from \cite[Theorem 6.2]{bluhm2022}, which extends the decay of correlations already known from Araki's seminal paper \cite{araki1969} to the case of Gibbs states on finite intervals.
The constant $K'$ here comes from the interaction dependent constant from the cited result combined with the (constant) norms of $Y^{-1}$, $Y_L$, $Y_R$ and the locality of $O$ reducing the distance between the supports of $O$ and $Y$. The constant $\alpha$ depends on the interaction.

We conclude by \cite[Theorem 4.16 (ii)]{perezgarcia2020} (see also \cite{araki1969}), which proves that expectations in the global Gibbs state are approximated by local ones, with an error decaying exponentially with the distance to the boundary, i.e.,
\begin{align}
|\varphi^\Lambda(O)-\OTI_{\beta}|\le \|O\|K'' \exp(-\delta \ell).
\end{align}
Again, $\delta$ and $K''$ depend on the interaction and the latter in addition on the locality of $O$.
This proves Theorem \ref{thm:maincomm1} in the case $D=1$, for all finite $\beta \geq 0$.

In higher dimension we expect the theorem to break down for arbitrary temperatures due to the existence of phase transitions.
Nevertheless, at sufficiently high temperatures, the two main ingredients, decay of correlations and local indistinguishability still hold.
A technical problem we are facing is however the (super)exponential growth of $e^{\beta W}$ since the support and norm of $W$ scale with $\ell^{D-1}$, where $\ell$ is the distance between $O$ and the boundary.
In 2D, however, we can still trade off this exponential growth against the decay rates by choosing sufficiently high temperatures.
We closely follow the proof of the 1D case.
Again we define 
$Y=(id_{\Lambda} \otimes \hat\omega_{\bdex \Lambda})(e^{-\beta W})\in\cA_\Lambda$.
This time, however, we can only bound $\|Y\|,\|Y^{-1}\|\le\exp(\beta \ell r\|h\|)$, where $r$ is a constant depending on the locality of the Hamiltonian. For the decay of correlations we can resort to the high temperature result in \cite[Theorem 2]{kliesch2014}.
Below a critical inverse temperature this gives
\begin{align}
|\omega(O)-\varphi^\Lambda(O)|&=\frac{|\varphi^\Lambda(OY)-\varphi^\Lambda(O)\varphi^\Lambda(Y)|}{\varphi^\Lambda(Y)}\\
&\le \|O\|\|Y\|\|Y^{-1}\| C\exp(-\gamma \ell)\\
&\le\|O\|C\exp((2\beta r\|h\|-\gamma)\ell)
\end{align}
where $C$ and $\gamma$ depend on the interaction and temperature.
Note that $\gamma$ grows unboundedly for $\beta\to0$.
By choosing $\beta$ sufficiently small such that $2\beta r\|h\|<\gamma$ this implies exponential decay.

Finally, local indistinguishability is proven in \cite[Corollary 2]{kliesch2014} again above the same critical temperature and with the same decay rate.
The formulation proves convergence of marginals in 1-norm distance but by choosing the region to be the support of $O$ the result for observables follows:
\begin{equation}
|\varphi^\Lambda(O)-\OTI_{\beta}|\le\|O\|K''\exp(-\gamma \ell)
\end{equation}
\end{proof}
\begin{remark}
    We remark that in the proof above, we do not make explicit use of the uniqueness of the $\beta$-KMS state for $0\leq \beta < \beta_1$. This uniqueness actually follows from the proof since it holds for any choice of KMS-state and any observable.
The used results on decay of correlations based on cluster expansions are in fact the same techniques used to prove the uniqueness of the KMS-state at sufficiently high temperature.
\end{remark}

\section{Numerical experiments}
\label{sec:numerics}
We consider the one-dimensional transverse field Ising model
\begin{equation}
\label{eq:tfising}
H = -\sum_{i \in \ZZ} Z_i Z_{i+1} + g X_i
\end{equation}
where $X_i$ and $Z_i$ are the usual Pauli matrices at site $i$ and $g \geq 0$. The model is exactly solvable \cite{pfeuty}, and the ground state exhibits a phase transition at $g=1$. For $g < 1$, the ground space is degenerate, and for $g > 1$, the ground state is unique. The model is gapped for $g \neq 1$, and gapless at $g=1$. There is no phase transition with respect to the temperature as this is a one-dimensional model.

Let $\KMS_{g,\beta}$ be the set of translation-invariant equilibrium states for the parameter $g$ of the Hamiltonian, and at inverse temperature $\beta$. For $g<1$, $|\KMS_{g,\infty}| > 1$, while for $g>1$ or $\beta < \infty$ we have $|\KMS_{g,\beta}|=1$.

We focus on two observables, namely the $Z$-magnetization per site, i.e., $O = Z_0$, and the nearest-neighbour correlation function $O = Z_0 Z_1$.

\paragraph{Magnetization} When $g=0$, the model is a simple one-dimensional Ising model, and there are two translation-invariant ground states, namely the all-up and all-down states. This means that at $g=0$ the magnetization can take any value between $-1$ and $+1$. For general $g \geq 0$, it has been shown that \cite[Eq. (3.12a)]{pfeuty}
\begin{equation}
\{\omega(Z_0) : \omega \in \KMS_{g,\infty}\} = [-M_z(g),M_z(g)]
\end{equation}
where
\begin{equation}
M_z(g) = \begin{cases}
(1-g^2)^{1/8} & \text{if } g < 1\\
0 & \text{if } g \geq 1.
\end{cases}
\end{equation}

Figure \ref{fig:tfising1dmagnetization} (left) shows the true values of the magnetization as a function of $g$, along with the bounds obtained from the convex relaxation \eqref{eq:opt2}. The relaxations were solved respectively with a 4-site and 5-site region $\Lambda$. (Note that the semidefinite program for a $L$-site region involves a positivity constraint of a matrix of size $4^{L-1}$.) We see as expected that the bounds with $L=5$ are tighter than the bounds for $L=4$. We also observe that the bounds seem very accurate away from the critical point $g=1$.

In Figure \ref{fig:tfising1dmagnetization} (right) we plot the lower and upper bounds on the magnetization at positive temperature $\beta=1$. Note that for positive temperature, there is a unique thermal equilibrium state and the average magnetization is 0 for all $g \geq 0$.

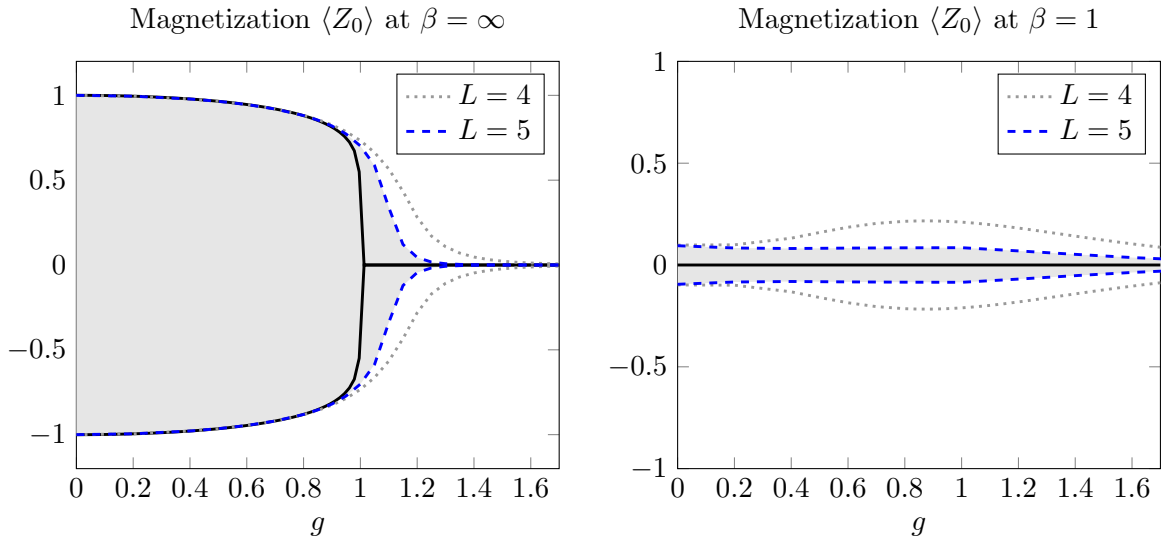
\begin{figure}[h]
\centering
\begin{tikzpicture}
\begin{axis}[xlabel={$g$},title={Magnetization $\<Z_0\>$ at $\beta=\infty$},yticklabel style={/pgf/number format/fixed},grid=none,width=8cm,height=7cm,xmin=0,xmax=1.7,legend entries={$L=4$,$L=5$},legend pos=north east,legend cell align={left}]
\addplot [color=black, domain=0.0:1.7, samples=100, line width=1.1, forget plot]{max(0,1-x^2)^(1/8)};
\addplot [color=black, domain=0.0:1.7, samples=100, line width=1.1, forget plot]{-max(0,1-x^2)^(1/8)};
\addplot [color=lightgrey, dotted, line width=1.1] table[x=g, y=Mz]{tfising1d-Mz-beta=inf-L=4.tex};
\addplot [color=lightgrey, dotted, line width=1.1, forget plot] table[x=g, y expr={-\thisrow{Mz}}]{tfising1d-Mz-beta=inf-L=4.tex}; 
\addplot [color=blue, name path=Zub, dashed, line width=1.1] table[x=g, y=Mz]{tfising1d-Mz-beta=inf-L=5.tex};
\addplot [color=blue, name path=Zlb, dashed, line width=1.1, forget plot] table[x=g, y expr={-\thisrow{Mz}}]{tfising1d-Mz-beta=inf-L=5.tex}; 
\addplot[gray!20] fill between[of=Zub and Zlb];
\end{axis}
\end{tikzpicture}
\quad
\begin{tikzpicture}
\begin{axis}[xlabel={$g$},title={Magnetization $\<Z_0\>$ at $\beta=1$},yticklabel style={/pgf/number format/fixed},grid=none,width=8cm,height=7cm,xmin=0,xmax=1.7,ymin=-1,ymax=1,legend entries={$L=4$,$L=5$},legend pos=north east,legend cell align={left}]
\addplot [color=black, domain=0.0:1.7, samples=100, line width=1.1, forget plot]{0};
\addplot [color=lightgrey, name path=Zub4, dotted, line width=1.1] table[x=g, y=Mz]{tfising1d-Mz-beta=1-L=4.tex};
\addplot [color=lightgrey, name path=Zlb4, dotted, line width=1.1, forget plot] table[x=g, y expr={-\thisrow{Mz}}]{tfising1d-Mz-beta=1-L=4.tex};
\addplot [color=blue, name path=Zub5, dashed, line width=1.1] table[x=g, y=Mz]{tfising1d-Mz-beta=1-L=5.tex};
\addplot [color=blue, name path=Zlb5, dashed, line width=1.1, forget plot] table[x=g, y expr={-\thisrow{Mz}}]{tfising1d-Mz-beta=1-L=5.tex};
\addplot[gray!20] fill between[of=Zub5 and Zlb5];
\end{axis}
\end{tikzpicture}
\caption{Upper and lower bounds on magnetization $\<Z_0\>$ for thermal equilibrium states ($\beta=\infty$ and $\beta=1$) of the quantum Ising model \eqref{eq:tfising}, as a function of $g$. The shaded region is the interval corresponding to $L=5$.}
\label{fig:tfising1dmagnetization}
\end{figure}

\paragraph{Spin-spin correlation function} Next, we consider the nearest-neighbour correlation function $O=Z_0Z_1$. The exact value for all choices of $g$ and $\beta \in [0,+\infty]$ is given by \cite[Eq. (10.58)]{sachdev_2011}:\footnote{Here the interval reduces to a single point, even when the set of equilibrium  states is degenerate for $g < 1$ and $\beta=\infty$.}
\begin{equation}
\<Z_0 Z_1\>_{g,\beta} = \frac{1}{\pi} \int_{0}^{\pi} \frac{1+g\cos k}{\sqrt{1+2g\cos k+g^2}} \tanh(J \beta \sqrt{1+2g\cos k+g^2}) dk.
\end{equation}

Figure \ref{fig:tfising1dcorrz} compares the bounds obtained from the convex relaxations with the true value. For this observable we see a bigger gap with the true value, however this gap seems to decay quickly as $g\gg1$ away from the critical region. We also notice that at $\beta=\infty$ (ground state), the upper bounds are much closer to the true value than the lower bounds, especially for $g<1$.

\begin{figure}[h!]
\centering
\begin{tikzpicture}
\begin{axis}[xlabel={$g$},title={Correlation $\<Z_0 Z_1\>$ at $\beta=\infty$},yticklabel style={/pgf/number format/fixed},grid=none,width=8cm,height=7cm,xmin=0.1,xmax=1.9,legend entries={$L=4$,$L=5$}]
\addplot [color=black, line width=1.1, forget plot] table[x=g, y=ZZ]{tfising1d-ZZ-beta=inf-true.tex};
\addplot [color=lightgrey, dotted, line width=1.1] table[x=g, y=ZZub]{tfising1d-ZZ-beta=inf-L=4.tex};
\addplot [color=lightgrey, dotted, line width=1.1, forget plot] table[x=g, y=ZZlb]{tfising1d-ZZ-beta=inf-L=4.tex};
\addplot [color=blue, name path=ZZub, dashed, line width=1.1] table[x=g, y=ZZub]{tfising1d-ZZ-beta=inf-L=5.tex};
\addplot [color=blue, name path=ZZlb, dashed, line width=1.1, forget plot] table[x=g, y=ZZlb]{tfising1d-ZZ-beta=inf-L=5.tex};
\addplot[gray!20] fill between[of=ZZub and ZZlb];
\end{axis}
\end{tikzpicture}
\quad
\begin{tikzpicture}
\begin{axis}[xlabel={$g$},title={Correlation $\<Z_0 Z_1\>$ at $\beta=1$},yticklabel style={/pgf/number format/fixed},grid=none,width=8cm,height=7cm,xmin=0.1,xmax=1.9,legend entries={$L=4$,$L=5$}]
\addplot [color=black, line width=1.1, forget plot] table[x=g, y=ZZ]{tfising1d-ZZ-beta=1-true.tex};
\addplot [color=lightgrey, name path=ZZub4, dotted, line width=1.1, forget plot] table[x=g, y=ZZub]{tfising1d-ZZ-beta=1-L=4.tex};
\addplot [color=lightgrey, name path=ZZlb4, dotted, line width=1.1] table[x=g, y=ZZlb]{tfising1d-ZZ-beta=1-L=4.tex};
\addplot [color=blue, name path=ZZub5, dashed, line width=1.1, forget plot] table[x=g, y=ZZub]{tfising1d-ZZ-beta=1-L=5.tex};
\addplot [color=blue, name path=ZZlb5, dashed, line width=1.1] table[x=g, y=ZZlb]{tfising1d-ZZ-beta=1-L=5.tex};
\addplot[gray!20] fill between[of=ZZub5 and ZZlb5];
\end{axis}
\end{tikzpicture}
\caption{Upper and lower bounds on correlation $\<Z_0 Z_1\>$ for thermal equilibrium states ($\beta=\infty$ and $\beta=1$) of the quantum Ising model \eqref{eq:tfising}, as a function of $g$. The shaded region is the interval corresponding to $L=5$.}
\label{fig:tfising1dcorrz}
\end{figure}
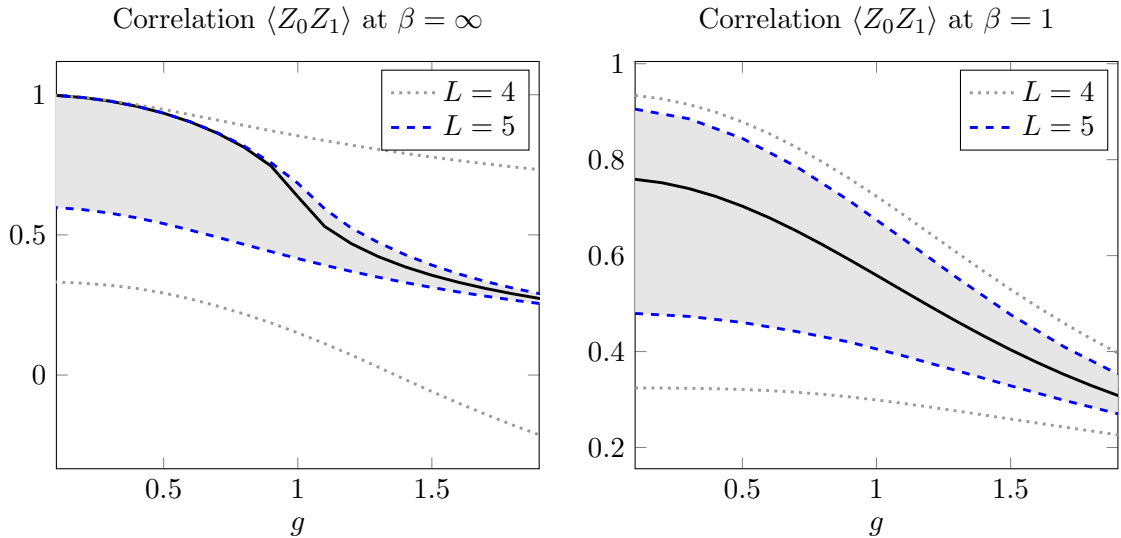

For $\beta < \infty$ the optimization problems involving the matrix relative entropy were solved using the semidefinite programming approximations \cite{logapprox}. We used the Mosek interior-point solver \cite{mosek} for all the computations, except for the case $L=5,\beta < \infty$ where we used the first-order solver SCS \cite{scs}.

\section*{Acknowledgments}
We thank Daniel Stilck Fran\c ca for helpful discussions.
We thank David P\'erez-Garc\'ia for pointing out the necessity of using subnets in the proof of Theorem~\ref{thm:main1}.
HF acknowledges funding from UK Research and Innovation (UKRI) under the UK
government’s Horizon Europe funding guarantee EP/X032051/1. OF acknowledges funding by the European Research Council (ERC Grant AlgoQIP, Agreement No. 851716) as well as 
by the European Union’s Horizon 2020 within the QuantERA II Programme under Grant VERIqTAS Agreement No 101017733. SOS acknowledges support from the UK Engineering and Physical Sciences Research Council (EPSRC) under grant number EP/W524141/1.

\newpage
\clearpage

\bibliography{refs}
\bibliographystyle{naturemag}
\clearpage
\appendix

\section{Proof of Theorem \ref{thm:matrixEEB}}
\label{sec:matrixEEBgeneral}

Our proof is analogous to \cite[Theorem 5.3.15]{brbook}.
Let $\{a_1,\ldots,a_m\}$ be a finite subset of $\cAloc$.
Let $(\mathcal{H},\pi,\Omega)$ be the cyclic representation for $\omega$ of $\cA$ \cite[Corollary 2.3.16]{brbook} and
\begin{equation}
U(t)=\int_{-\infty}^\infty e^{-ipt}dE(p)
\end{equation}
the canonical unitary group implementing the time evolution $\alpha_t$ \cite[Corollary 2.3.17]{brbook} in its spectral decomposition.
Analogously to \cite[page 88]{brbook}, we can define the following positive $m\times m$ Hermitian matrix-valued measures $\mu$ and $\nu$.
\begin{align}
[\mu(\hat{f})]_{kl}& = (2\pi)^{-1/2} \int_{-\infty}^{\infty} dt f(t) \omega(a_k^* \alpha_t(a_l))=\int_{-\infty}^\infty\langle\pi(a_k)\Omega,dE(q)\pi(a_l)\Omega\rangle\hat{f}(q)dq\\
[\nu(\hat{f})]_{kl} &= (2\pi)^{-1/2} \int_{-\infty}^{\infty} dt f(t) \omega(\alpha_t(a_l) a_k^*)=\int_{-\infty}^\infty\langle\pi(a_k)\Omega,dE(-q)\pi(a_l)\Omega\rangle\hat{f}(q)dq
\end{align}

Here, the first part of each definition assumes a compactly supported and infinitely differentiable function $\hat f$ and its inverse Fourier transform
\begin{equation}
f(z)=(2\pi)^{-1/2}\int_{\infty}^\infty \hat f(p)e^{ipz}dp,
\end{equation}
whereas the second part of the definition of the measures defines their continuation to arbitrary bounded continuous functions $\hat f$.

We see from the second part that $[\mu(1)]_{kl}=\omega(a_k^*a_l) = A_{kl}$ and $\nu(1)=B$.
Furthermore, we show that $d\nu(p) = e^{p\beta} d\mu(p)$ following the analogous statement in \cite[Proposition 5.3.14]{brbook}.
We do this by computing for a compactly supported infinitely differentiable function $\hat f$
\begin{align}
    [\mu(\hat f)]_{kl}&=\int_{\infty}^\infty dt f(t)\omega(a_k^*\alpha_t(a_l))\\
    &=\int_{\infty}^\infty dt f(t+i\beta) \omega(\alpha_t(a_l)a_k^*)=[\nu(k_\beta\hat f)]_{kl}
\end{align}
where $k_{\beta}(p)=e^{-\beta p}$, 
also using the characterization of KMS-states from \cite[Proposition 5.3.12]{brbook}.

We now prove the inequality $D_{op}(A\|B) \nsd \beta C$.
Let $\gamma$ be the positive real-valued measure defined by $d\gamma(p) =  \tr[d\mu(p)]$. Since the trace is faithful, the measure $\mu$ has a matrix-valued density with respect to $\gamma$, i.e., we can write $d\mu(p) = \tilde{A}(p) d\gamma(p)$ where $\tilde{A}:\RR \to \H^m_{+}$ (see e.g., \cite[Section 4]{farenick2007jensen}). Then we can write
\begin{equation}
\begin{aligned}
D_{op}(A\|B) &= D_{op}\left(\int 1 d\mu(p) \| \int 1 d\nu(p)\right)\\
             &= D_{op}\left(\int \tilde{A}(p) d\gamma(p) \| \int  e^{p\beta} \tilde{A}(p) d\gamma(p)\right)\\
             &\nsd \int D_{op}(\tilde{A}(p)\|e^{p\beta} \tilde{A}(p)) d\gamma(p)\\
             &= \beta \int -p \tilde{A}(p) d\gamma(p) = \beta \int -p d\mu(p) = C.
\end{aligned}
\end{equation}
\clearpage

\end{document}